%% file: esa-guided.tex
\documentclass{llncs}

\usepackage[utf8]   {inputenc}
\usepackage	    {hyperref}
\usepackage         {amsmath} % binom
\usepackage         {amssymb} % checkmark
\usepackage         {complexity}
\usepackage         {mathrsfs}
\usepackage[numbers]{natbib} % sectionbib: to avoid a page break before references
\usepackage         {paralist}
\usepackage         {tikz}
\usetikzlibrary {shapes}
\usepackage [backgroundcolor=white,linecolor=black]    {todonotes}

\usepackage         {xparse}

\input{notation}

\usepackage{cleveref}
\Crefname{observation}{Observation}{Observations}
\Crefname{hypothese}{Assumption}{Assumptions}

\title{Sorting With Forbidden Intermediates}
\author{Carlo Comin\inst{1,2} \and Anthony Labarre\inst{1} \and Romeo Rizzi\inst{3} \and Stéphane Vialette\inst{1}}
\authorrunning{Carlo Comin et al.}
\institute{
Universit{\'e} Paris-Est, LIGM (UMR 8049), UPEM, CNRS, ESIEE, ENPC, F-77454, Marne-la-Vall{\'e}e, France\\
\email{\{Anthony.Labarre,Stephane.Vialette\}@u-pem.fr}
\and 
Department of Mathematics, University of Trento, Italy \\
\and
Department of Computer Science, University of Verona, Italy \\
\email{Carlo.Comin@unitn.it}, \email{Romeo.Rizzi@univr.it}
}
\begin{document}

\maketitle

\begin{abstract}
A wide range of applications, most notably in comparative genomics, 
involve the computation of a shortest sorting sequence of operations for a given permutation,
where the set of allowed operations is fixed beforehand. 
Such sequences are useful for instance when reconstructing potential scenarios of evolution between species,
or when trying to assess their similarity. We revisit those problems by adding a new constraint on the sequences to be computed:
they must avoid a given set of \emph{forbidden intermediates},
which correspond to species that cannot exist because the mutations that would be involved in their creation are lethal.
We initiate this study by focusing on the case where the only mutations 
that can occur are exchanges of any two elements in the permutations,
and give a polynomial time algorithm for solving that problem when the permutation to sort is an involution.
\end{abstract}
\textbf key
Guided Sorting, Lethal Mutations, Forbidden Vertices, Permutation Sorting, Hypercube Graphs, st-Connectivity.

\thispagestyle{plain}  % number page 1

\input{Section1-Introduction}
\input{Section2-Background}
\input{Section3-Algorithm}

% =============================================================================
\section{Conclusion}\label{sect:conclusion}
% =============================================================================

%FORMER ESA VERSION:
With the intention of integrating more biologically relevant constraints into classical genome rearrangement problems,
we introduced in this paper the {\sc guided sorting} problem.
We broadly define it as the problem of transforming two genomes into one another using as few operations as possible
from a given fixed set of allowed operations while avoiding a set of nonviable genomes.
We gave a polynomial time algorithm for solving this problem in the case where genomes are represented by permutations,
under the assumptions that
 \begin{inparaenum}[1)]
     \item permutations can only be modified by exchanging any two elements,
     \item the sequence to seek must be optimal, and
     \item the permutation to sort is an involution.
 \end{inparaenum}

 % Many questions remain open, most notably that of the complexity of the {\sc guided sorting} problem, whether under the constraints we imposed or in a more general setting. Other structures could indeed be used to represent genomes, different operations could be allowed, and other constraints could be set on the sequences to be computed (e.g., not requiring that they be optimal), which would all yield different variants whose complexity status may vary.
% Aside from complexity issues, future work shall also focus on extending the approach we proposed to other families of instances of the {\sc guided sorting} problem, and identifying other classes that can be solved in polynomial time.

%ISAAC Version
%We introduced in this paper the {\sc guided sorting} problem, which we broadly define as the problem of transforming two structures into one another using as few moves as possible from a predefined set of allowed operations while avoiding a set of forbidden intermediate configurations.
%We gave a polynomial time algorithm for solving the problem on permutations when%, under the simplifying assumptions that
%\begin{inparaenum}[(1)]
%    \item only exchanges of any two elements are allowed,
%\end{inparaenum}
%    \item the wanted sequence is optimal, and
%    \item the permutation to sort is an involution.
%\end{inparaenum}

Many questions remain open, most notably that of the computational complexity of the {\sc guided sorting} problem, whether under
% the constraints we imposed
assumptions (1) and (2) or in a more general setting (i.e., using structures other than permutations, operations other than exchanges,
or allowing sequences to be ``as short as possible'' instead of optimal).
One could also investigate ``implicit'' representations for the set of forbidden intermediate permutations,
e.g. all permutations that avoid a given (set of) pattern(s), or that belong to a specific conjugacy class.
Aside from complexity issues, future work shall also focus on extending the approach
we proposed to other families of instances of the {\sc guided sorting} problem,
and identifying other tractable (or intractable) cases or variants of it;
for instance, we plan to extend our algorithmic results to the family of graphs satisfying
the \emph{shadow-matching}~\cite{LS04} condition.%, and to be able to handle \emph{pairs} of forbidden vertices as well~\cite{KSG73}

\bibliographystyle{mynatstyle}
\bibliography{rearr-corr}

\label{BibliographyPage}
% \thispagestyle{fancy} % TODO REMOVE FOR SUBMISSION

% -----------------------------------------------------------------------------
% APPENDICES ------------------------------------------------------------------
% -----------------------------------------------------------------------------
% -----------------------------------------------------------------------------
%\clearpage
%\appendix\label{app:guided-reduction}
%\section{\appendixname: Instances of {\sc guided sorting} that Reduce to \mainproblem}
%\label{sec:using-hypercubes}
%\input{GuidedSorting-Reduction}
%------------------------------------------------------------------------------
%\clearpage
%\label{app:omitted-proofs}
%\section{\appendixname: Omitted Proofs}
%\label{sect:omitted_proofs}
%\input{Omitted-Proofs}
%------------------------------------------------------------------------------
\end{document}

%% file: notation.tex
\newcommand{\nonviableset}            {\ensuremath{\mathcal{F}}}
\newcommand{\permutation}         [1] {\ensuremath{\langle}\mbox{#1}\ensuremath{\rangle}}

\newcommand{\relevantset}         [2] {\ensuremath{R_{#1}(#2)}}

\newcommand{\Sym}                 [1] {\ensuremath{\mathfrak{S}_{#1}}}

\newcommand{\ie}{i.e.\ }
\newcommand{\eg}{e.g.\ }

\newcommand{\mainproblem}{{\sc Hy-stCon}}

\newcommand{\uN}[1]{N(#1)}
\newcommand{\inN}[1]{N^{\texttt{in}}(#1)}
\newcommand{\outN}[1]{N^{\texttt{out}}(#1)}

\newcommand{\suchthat}{\;\ifnum\currentgrouptype=16 \middle\fi|\;}

\newcommand{\N}{\mathbb{N}\xspace}

\def\X{{\cal X}}

\def\H{{\cal H}}
\def\F{{\cal F}}
\def\G{{\cal G}}
\def\K{{\cal K}}
\def\M{{\cal M}}

\def\S{{\cal S}}
\def\R{{\cal R}}
\def\T{{\cal T}}
\def\P{{\cal P}}
\def\QQ{{\cal Q}}

\newcommand{\figref}[1]{Fig.~\ref{#1}}

\makeatletter
\providecommand*{\cupdot}{%
  \mathbin{%
    \mathpalette\@cupdot{}%
  }%
}
\newcommand*{\@cupdot}[2]{%
  \ooalign{%
    $\m@th#1\cup$\cr
    \sbox0{$#1\cup$}%
    \dimen@=\ht0 %
    \sbox0{$\m@th#1\cdot$}%
    \advance\dimen@ by -\ht0 %
    \dimen@=.5\dimen@
    \hidewidth\raise\dimen@\box0\hidewidth
  }%
}

\newcommand\restr[2]{{% we make the whole thing an ordinary symbol
  \left.\kern-\nulldelimiterspace % automatically resize the bar with \right
  #1 % the function
  \vphantom{\big|} % pretend it's a little taller at normal size
  \right|_{#2} % this is the delimiter
}}

\usepackage[ruled,vlined,linesnumbered]{algorithm2e}
\usepackage{subfig}
\SetCommentSty{textsf}
\SetKwRepeat{DoWhile}{do}{while}
\SetAlFnt{\small}
\makeatletter
\newcommand{\removelatexerror}{\let\@latex@error\@gobble}
\makeatother

\let\oldnl\nl
\newcommand{\nonl}{\renewcommand{\nl}{\let\nl\oldnl}}

\usepackage{tikz-cd}%Front end per PGF, package di disegno

\tikzset{
  shorten/.style={/tikz/shorten >={#1},/tikz/shorten <={#1}}}

\usepackage{tikz-qtree}
\usetikzlibrary{calc,positioning,fit}%calc e positioning permettono il posizionamento avanzato
\usetikzlibrary{shapes,shapes.multipart,shapes.arrows}
\usetikzlibrary{decorations,decorations.pathmorphing,decorations.pathreplacing,decorations.markings,decorations.shapes}
\usetikzlibrary{arrows}%se si usano le frecce diverse
\usetikzlibrary{fit,backgrounds}
\usetikzlibrary{plotmarks}
\tikzstyle{node}=[circle,draw,inner sep=2pt,transform shape,minimum size=1.75em]

%% file: Section1-Introduction.tex
% =============================================================================
\section{Introduction}
% =============================================================================

Computing distances between permutations, or sequences of operations that transform them into one another,
are two generic problems that arise in a wide range of applications,
including comparative genomics~\cite{Fertin2009}, ranking~\cite{Deza:1998},
and interconnection network design~\cite{Lakshmivarahan1993}.
Those problems are well-known to reduce to constrained sorting problems of the following form:
given a permutation $\pi$ and a set $S$ of allowed operations,
find a sequence of elements from $S$ that sorts $\pi$ and is as short as possible.
In the context of comparative genomics, the sequence to be reconstructed
yields a possible scenario of evolution between the genomes represented by $\pi$ and the target
identity permutation $\iota$, where all permutations obtained inbetween are successive descendants of
$\pi$ (and ancestors of $\iota$). The many possible choices that exist for $S$,
as well as other constraints or cost functions with which they can be combined,
have given rise to a tremendous number of variants whose algorithmic and mathematical
aspects have now been studied for decades~\cite{Fertin2009}.
Specific issues that biologists feel need to be addressed to improve
the applicability of these results in a biological context include:
% \begin{inparaenum}
\begin{enumerate}%[1)]
    \item the oversimplicity of the model (permutations do not take duplications into account),
    \item the rigid definition of allowed operations, which fails to capture the complexity of evolution, and
    \item the complexity of the resulting problems,
where algorithmic hardness results abound even for deceivingly simple problems.
% \end{inparaenum}
\end{enumerate}
A large body of work has been devoted to addressing those issues, namely by proposing richer models for genomes,
encompassing several operations with different weights~\cite{Fertin2009}.
Some approaches for increasing the reliability of rearrangement methods
by adding additional biologically motivated constraints have been investigated
% (see for example
(for instance,
% \cite{Bergeron:Blanchette:Chateau:Chauve:WABI:2004} for conserved intervals,
\citet{Bergeron:Blanchette:Chateau:Chauve:WABI:2004} consider \emph{conserved intervals},
% \cite{Figeac:Varre:WABI:2004} for restricting the set of allowed inversions and
\citet{Figeac:Varre:WABI:2004} restrict the set of allowed inversions and
% \cite{Berard:Bergeron:Chauve:RECOMB:2004}
\citet{Berard:Bergeron:Chauve:RECOMB:2004} take into account
% for preserving
the number of inversions in the wanted scenario which commute with all \emph{common intervals}).
However, another critical issue has apparently been overlooked: to the best of our knowledge,
no model takes into account the fact that the solutions it produces may involve allele mutations that are lethal
%fatal
to the organism on which they act.
\emph{Lethals} are usually a result of mutations in genes
that are essential to growth or development~\cite{Gluecksohn-Waelsch63};
they have been known to occur for more than a century~\cite{Cuenot1905},
% as they were first discovered by
dating back to the works of
Cu\'enot in 1905
% while studying the inheritance of coat colour in mice.
who was studying the inheritance of coat colour in mice.
As a consequence, solutions that may be perfectly valid from a
mathematical point of view should nonetheless be rejected on the grounds that some of
the intermediate ancestors they produce are nonviable
and can therefore not have had any descendants.
We revisit the family of problems mentioned above by adding a natural constraint which, as far as we know,
has not been previously considered in this
form (see e.g.~\cite{Bergeron:Blanchette:Chateau:Chauve:WABI:2004,Figeac:Varre:WABI:2004,Berard:Bergeron:Chauve:RECOMB:2004} for connected attempts): namely, the presence of a set of forbidden intermediate permutations,
which the sorting sequence that we seek must avoid. We refer to this family of problems
as {\sc guided sorting} problems, since they take additional guidance into account.
In this paper, we focus our study on the case where
only \emph{exchanges} (i.e., algebraic transpositions) are allowed;
furthermore, we simplify the problem by demanding that the solutions we seek
be \emph{optimal} in the sense that no shorter sorting sequence of exchanges exists even
when no intermediate permutation is forbidden.
We choose to focus on exchanges because of their connection to the underlying
\emph{disjoint cycle structure} of permutations, which plays an important role
in many related sorting problems where a similar cycle-based approach,
using this time the ubiquitous \emph{breakpoint graph}, has proved extremely fruitful~\cite{Labarre2013}.
Therefore, we believe that progress on this particular
variant will be helpful when attempting to solve related variants based on more complex operations.

\subsection{Contribution}
Our main contribution in this work is a polynomial time algorithm for solving {\sc guided sorting}
by exchanges when the permutation to sort is an \emph{involution}.
We show that, in that specific case,
the space of all feasible sorting sequences admits a suitable
description in terms of directed $(s,t)$-paths in hypercube graphs.
We achieve this result by reducing {\sc guided sorting}
to the problem of finding directed $(s,t)$-paths that avoid a prescribed set
$\F\subseteq V$ of \emph{forbidden vertices}.
Our main contribution, therefore, consists in solving this latter problem in time polynomial in just the encoding length of $\F$, 
if $G$ is constrained to be a \emph{hypercube} graph;
which is a novel algorithmic result that may be of independent interest.
Specific properties that will be described later
on~\cite{GoldsLehmanRon2001, LehmanRon2001} allow us to avoid the full construction of that graph,
which would lead to an exponential time algorithm.

\subsection{Related Works}
We should mention that constrained variants of the $(s,t)$-connectivity
problem have been studied already to some extent.
For instance, in the early '70s, motivated by some problems in the field of automatic software testing and validation,
\citet{KSG73} introduced the \emph{path avoiding forbidden pairs} problem, namely,
that of finding a directed $(s,t)$-path in a graph $G=(V,E)$ that contains at most one vertex
from each pair in a prescribed set $\mathcal{P}\subseteq V\times V$ of \emph{forbidden pairs} of vertices.
\citet{GMO76} proved that the problem is \NP-complete on DAGs.
A number of special cases were shown to admit polynomial time algorithms,
\eg \citet{Yin97} studied the problem in directed graphs under a \emph{skew-symmetry} condition.
However, the involved techniques and the related results do not extend to our problem,
for which we are aware of no previously known algorithm that runs in time polynomial in just the encoding length of $\F$.

A preliminary version of this article appeared in the proceedings of the
\emph{3rd International Conference on Algorithms for Computational Biology (AlCoB 2016)}, see~\cite{Comin2016}.
Here, the previous results are improved and the presentation is extended:

\emph{(1)} The time complexity of Algorithm~\ref{ALGO:solve} is improved by a factor of $d_{S,T}\cdot n$ (see Theorem~\ref{thm:main} for the actual time bound).

\emph{(2)} Subsection~\ref{sect:VertexDisjointPaths} is extended by presenting Algorithm~\ref{ALGO:vertex-disjoint-paths}
  plus all the details of its correctness and running time analysis.

\emph{(3)} Subsection~\ref{sect:st-dPaths} is extended by including a detailed correctness and complexity analysis of Algorithm~\ref{ALGO:solve}.

\emph{(4)} \figref{fig:auxiliary_network}, \figref{fig:undirected_bipartite} and \figref{fig:undirected_bipartite_vertex_cover}
  have been added to support some of the more technical constructions with an illustration.
\subsection{Organization}
The remainder of the article is organized as follows.
Section~\ref{sec:background} provides some background notions and notation on which the rest of this work relies.
The main contribution is offered in Section~\ref{sect:main_sect}. In Subsection~\ref{sct:problem_formulation},
the problem {\mainproblem} is formulated. The reduction from {\sc guided sorting} for exchanges (and adjacent exchanges) to {\mainproblem} is
offered in Subsection~\ref{sec:guided-sorting-for-exchanges-details}~(and \ref{sec:guided-sorting-for-adjacent-exchanges-details}, respectively).
The formal statement of our main algorithmic contribution is detailed in Subsection~\ref{sct:main_result}.
Next, Subsection~\ref{sect:VertexDisjointPaths} concerns the specific properties~\cite{GoldsLehmanRon2001, LehmanRon2001} that allow us to avoid the full construction of the hypercube search space.
Subsection~\ref{sect:st-dPaths} presents the polynomial-time algorithm for solving \mainproblem.
In Subsection~\ref{sect:remark_dec_search}, it is shown how to speed up the algorithm in the case in which one is interested just in the decision task of \mainproblem.
The correctness analysis of the main algorithm is carried on in Subsection~\ref{subsect:correctness}, while the complexity is analyzed in Subsection~\ref{subsect:complexity}.
% The article ends in Section~\ref{sect:conclusion}.
We conclude in Section~\ref{sect:conclusion} with a discussion of several open problems.

%% file: Section2-Background.tex
% =============================================================================
\section{Background and Notation}\label{sec:background}
% =============================================================================

% \todo[inline,author=Anthony]{because we use $n$ in the definition of {\sc Hy-stCon}, let's switch to another letter for permutations; I'll use $k$ here but feel free to propose anything better}

We use the notation $\pi=\langle\pi_1\ \pi_2\ \cdots\ \pi_k\rangle$ when viewing permutations as sequences, i.e. $\pi_i=\pi(i)$ for $i\in[k]=\{1, 2, \ldots, k\}$. 
Our aim is to sort a given permutation $\pi$, i.e. to transform it into the identity permutation $\iota=\langle 1\ 2\ \cdots\ k\rangle$, using a predefined set of allowed operations 
specified as a generating set $S$ of the symmetric group $\Sym{k}$. We seek a sorting sequence that uses only elements from $S$ and:
\begin{enumerate}
    \item \emph{avoids} a given set $\nonviableset$ of \emph{forbidden permutations}, 
i.e. no intermediary permutation produced by applying the operations specified 
by the sorting sequence belongs to $\nonviableset$, and
    \item is \emph{optimal}, 
i.e. no shorter sorting sequence exists for $\pi$ even if $\nonviableset=\emptyset$. 
\end{enumerate}
We refer to the general problem of finding a sorting sequence under these constraints as {\sc guided sorting}, 
and restrict in this paper the allowed operations to \emph{exchanges} of any two elements (\ie \emph{algebraic transpositions}). 
For instance, let $\pi=\permutation{2 3 1 4}$ and $\nonviableset=\{\permutation{1 3 2 4}, \permutation{3 2 1 4}\}$. 
Then $\permutation{2 3 1 4}\mapsto\permutation{2 1 3 4}\mapsto\permutation{1 2 3 4}$ 
is a valid solution since it is optimal and avoids $\nonviableset$, but neither 
$\permutation{2 3 1 4}\mapsto\permutation{4 3 1 2}\mapsto\permutation{4 3 2 1}
\mapsto\permutation{4 2 3 1}\mapsto\permutation{1 2 3 4}$ nor 
$\permutation{2 3 1 4}\mapsto\permutation{1 3 2 4}\mapsto\permutation{1 2 3 4}$ can be accepted:
 the former is too long, and the latter does not avoid $\nonviableset$. 

We use standard notions and notation from graph theory (see e.g. \citet{Diestel2005} for undefined concepts), 
using $\{u, v\}$ (resp. $(u, v)$) to denote the edge (resp. arc) between vertices 
$u$ and $v$ of an undirected (resp. directed) graph $G=(V,E)$. 
All graphs we consider are \emph{simple}: they contain neither loops nor parallel edges / arcs.
If $\nonviableset\subseteq V$, 
a directed path $\p=v_0v_1\cdots v_n$ \emph{avoids} $\nonviableset$ when $v_i\not\in\nonviableset$ for every $i$.
% If $\S\subseteq V$ and $\T\subseteq V$, we say that a directed path $\p$ \emph{goes from $\S$ to $\T$ in $G$} when $\p$ starts from some $s$ in $\S$ and ends at some $t$ in $\T$. 
If $\S\subseteq V$ and $\T\subseteq V$, we say that $\p$ \emph{goes from $\S$ to $\T$ in $G$} if $v_0\in\S$ and $v_n\in \T$. 
When $G$ is directed, 
we partition the neighbourhood $\uN{u}$ of a vertex $u$ into the sets $\outN{u} = \{ v\in V \mid (u,v) \in E \}$ and $\inN{u}= \{v\in V\mid (v,u) \in E \}$. 
Some of our graphs may be vertex-labelled, using any injective mapping $\ell:V\rightarrow \N$.
For any $n\in \N$, $\wp_n=\wp([n])$ denotes the power set of $[n]$.
The \emph{hypercube graph on ground set} $[n]$, denoted by $\H_n$, is the graph with
vertex set $\wp_n$ and in which the arc $(U, V)$ connects vertices $U,V\subseteq [n]$ if 
there exists some $q\in [n]$ such that $U=V\setminus \{q\}$. 
If $S,T\in\wp_n$ and $|S|\leq |T|$, then $d_{S,T}=|T|-|S|$ is the \emph{distance} between $S$ and $T$.
Finally, $\H_n^{(i)}$ denotes the family of all subsets of $\wp_n$ of size $i$.

%% file: Section3-Algorithm.tex
% =============================================================================
\section{Solving {\sc guided sorting} For Involutions}\label{sect:main_sect}
% =============================================================================

The \emph{Cayley graph} $\Gamma(\Sym{n}, S)$ of $\Sym{n}$ for a given generating set $S$ of $\Sym{k}$
contains a vertex for each permutation
in $\Sym{k}$ and an edge between any two permutations that can be obtained from one another using one element from $S$.
A na\"{\i}ve approach for solving any variant of the {\sc guided sorting} problem would build
the part of $\Gamma(\Sym{k}, S)$ that is needed (i.e. without the elements of $\nonviableset$),
then run a shortest path algorithm to compute an optimal sequence that avoids all elements of $\nonviableset$.
This is highly impractical, since the size of $\Gamma$ is exponential in $k$.

We describe in this section a polynomial time algorithm
% in the case of exchanges if $\pi$ is an \emph{involution},
for the case where $S$ is the set of all exchanges and $\pi$ is an \emph{involution},
i.e. a permutation such that for each $1\leq i\leq |\pi|$,
either $\pi_i=i$ or there exists an index $j$ such that $\pi_i=j$ and $\pi_j=i$.
From our point of view, involutions reduce to collections of disjoint pairs of elements that each need to be swapped by an exchange until we obtain the identity permutation,
and the only forbidden permutations that could be produced by an optimal sorting sequence are involutions whose pairs of unsorted elements all appear in $\pi$.
Therefore, we can reformulate our {\sc guided sorting} problem in that setting as that of finding a directed $(\pi, \iota)$-path in $\H_n$ that avoids all vertices in $\F$,
where the permutation to sort $\pi$ corresponds to the top vertex $[n]$ of $\H_n$ and the identity permutation $\iota$ corresponds to the bottom vertex $\emptyset$ of $\H_n$. We give more details on the reduction in \Cref{sec:guided-sorting-for-exchanges-details}.

%\todo[inline,author=Anthony]{I disagree: if $\pi$ has $k$ $2$-cycles (and $k\leq n/2$),
%we go from the top vertex $[k]$ of $H_k$ to the bottom vertex $\emptyset$ since $\iota$ has only $1$-cycles}
% THE FOLLOWING TODO HAS TO BE ADDRESSED, BUT WE HAVE NO TIME TO DO IT WITHOUT RISKS BEFORE THE DEADLINE
% \todo{isn't there a confusion here? I thought we mapped the $k$ $2$-cycles of $\pi$ onto $[k]$ and the $0$ cycles of $\iota$ onto $\emptyset$ (same thing for elements in $\nonviableset$), which seems more natural. But then there would be a lot of changes to make in the text ...}.
\input{Section3-Algorithm-ProblemFormulation}
\input{Section3-Algorithm-LehmanRon}
\input{Section3-Algorithm-Main}

%% file: Section3-Algorithm-ProblemFormulation.tex
\subsection{Problem Formulation}\label{sct:problem_formulation}
We shall focus on the following problem from here on.

% \todo[inline,author=Anthony]{because of the above TODO, I'm switching from $n$ to $k$ in the notation below}
% I changed again, it was simpler to change n to k in the background section for permutations

\smallskip
\framebox{
\begin{minipage}{11cm}
\textsc{Problem:} \mainproblem.

\rule{\textwidth}{.9pt}
\textsc{Input}:
the size $n\in \N$ of the underlying ground set $[n]$,
a family of \emph{forbidden vertices} $\nonviableset\subseteq \wp_n$,
a \emph{source} set $S\in \wp_n$ and a \emph{target} set $T\in \wp_n$.
% the size $k\in \N$ of the underlying ground set $[k]$,
% a family of \emph{forbidden vertices} $\nonviableset\subseteq \wp_k$,
% a \emph{source} set $S\in \wp_k$ and a \emph{target} set $T\in \wp_k$.
\rule{\textwidth}{0.5pt}
\textsc{Decision-Task:} Decide whether there exists a directed path $\p$ in $\H_n$ that goes
% \textsc{Decision-Task:} Decide whether there exists a directed path $\p$ in $\H_k$ that goes
from source $S$ to target $T$ avoiding $\nonviableset$; \\
\textsc{Search-Task:} Compute a directed path $\p$ in $\H_n$ that goes from source $S$ to target $T$
% \textsc{Search-Task:} Compute a directed path $\p$ in $\H_k$ that goes from source $S$ to target $T$
avoiding $\nonviableset$, provided that at least one such path exists.
\end{minipage}
}
\smallskip

We examine in this section specific instances of {\sc guided sorting} which can be solved through a
reduction to {\mainproblem}. We say that permutations that may occur in an optimal sorting sequence for a given permutation $\pi$ are \emph{relevant}, and all others are \emph{irrelevant}. The distinction will matter when sorting a particular permutation since, as we shall see, the structure of $\pi$ (however it is measured) will have implications on that of relevant permutations and will allow us to simplify the set of forbidden permutations by discarding irrelevant ones.
% A \emph{transposition} is a permutation that swaps any two elements of $[n]$, or equivalently,
% whose disjoint cycle decomposition consists of one cycle of length two and $n-1$ cycles of length $1$.
% It is customary to write a transposition using the shorthand notation $(i, j)$,
% which indicates the only two elements that change places.
% An \emph{adjacent transposition} is a transposition that involves consecutive elements,
% i.e. a transposition of the form $(i, i+1)$.
% The traditional sorting problems associated to transpositions and adjacent transpositions are
% well-known and can be solved in polynomial time by the following strategies:
% \begin{itemize}
%     \item for transpositions: each transposition can be shown to either increase or decrease the number of cycles by one.
% Given that the identity permutation has the largest number of cycles (namely, $n$),
% we repeatedly apply any transposition that splits a cycle of length $\geq 2$ until we reach $\iota$.
%     \item for adjacent transpositions: run the bubble sort algorithm,
% which eliminates exactly one pair of adjacent elements in the wrong order at each step.
% \end{itemize}
For a fixed set $S$ of operations, we let $\relevantset{S}{\pi}$ denote the set of permutations that are relevant to $\pi$. Undefined terms and unproven properties of permutations below are well-known, and details are in standard references, such as~\citet{DBLP:books/aw/Knuth73}.

\subsection{{\sc guided sorting} For Exchanges}\label{sec:guided-sorting-for-exchanges-details}

% \todo[inline]{actually: it seems that what's below works as long as all cycles of length $>1$ in
% $\pi$ and in all elements of $\nonviableset$ have the same length, doesn't it?}

Recall that every permutation $\pi$ in $\Sym{k}$ decomposes in a single way into
\emph{disjoint cycles} (up to the ordering of cycles and of elements within each cycle).
This decomposition corresponds to the cycle decomposition of the directed graph $G(\pi)=(V, A)$,
where $V=[k]$ and $A=\{(i, \pi_i)\ |\ 1\leq i \leq k\}$.
The \emph{length} of a cycle of $\pi$ is then simply the number of elements it contains, and the number of cycles of $\pi$ is denoted by $c(\pi)$.

The \emph{Cayley distance} of a permutation $\pi$ is the length of an optimal sorting sequence of exchanges for $\pi$, and its value is $|\pi|-c(\pi)$.
Therefore, when searching for an optimal sorting sequence, we may restrict our attention to exchanges that split a particular cycle into two smaller ones.

Let $(\pi, \nonviableset, S, K)$ be an instance of {\sc guided sorting} such that $S$
is the set of all exchanges and where the permutation $\pi$ to sort is an \emph{involution},
i.e. a permutation whose cycles have length at most two. It is customary to omit cycles of length $1$,
and to write a permutation $\pi=\langle\pi_1\ \pi_2\ \cdots\ \pi_k\rangle$ with $n$ cycles of length $2$ as $c_1c_2\cdots c_n$.
Since we are looking for an optimal sorting sequence, we may assume that all permutations in $\nonviableset$ are relevant,
which in this case means that every permutation $\phi$ in $\nonviableset$ is an involution and its $2$-cycles form a proper subset of those of $\pi$.
Our instance of {\sc guided sorting} then translates to the following instance of {\mainproblem}:
\begin{itemize}
    \item $\pi\mapsto [n]$ in the following way: $c_i\mapsto i$ for $1\leq i\leq n$;
    \item each permutation $\phi$ in $\nonviableset$ is mapped onto a subset of $[n]$ by replacing its cycles with the indices obtained in the first step;
let $\nonviableset'$ denote the collection of subsets of $[n]$ obtained by applying that mapping to each $\phi$ in $\nonviableset$.
\end{itemize}
The resulting {\mainproblem} instance is then $\langle [n], \emptyset, \nonviableset', n\rangle$,
and a solution to instance $(\pi, \nonviableset, S, K)$ of {\sc guided sorting} exists if and only if a solution
to instance $\langle [n], \emptyset, \nonviableset', n\rangle$ of {\mainproblem} exists;
the translation of the solution from the latter formulation to the former is straightforward.

% -----------------------------------------------------------------------------
\subsection{{\sc guided sorting} For Adjacent Exchanges}\label{sec:guided-sorting-for-adjacent-exchanges-details}

Recall that an \emph{inversion} in a permutation $\pi$ in $\Sym{k}$ is a pair $(\pi_i, \pi_j)$
with $1\leq i<j\leq k$ and $\pi_i>\pi_j$. Let $(\pi, \nonviableset, S, K)$ be an instance
of {\sc guided sorting} where $S$ is the set of all \emph{adjacent} exchanges, i.e. exchanges that act on consecutive positions.
It is well-known that in this case, any optimal sorting sequence for $\pi$ has length equal to the number of inversions of $\pi$,
which means that in the search for an optimal sorting sequence,
we may restrict our attention to adjacent exchanges that act on inversions that consist of adjacent elements.
% Anthony: NOT TRUE, COMMENTING:
%  and disregard
% \emph{fixed points} (i.e. positions $i$ such that $\pi_i=i$, which correspond to elements that are already where they should be in $\iota$).

Let us now assume that all $n$ inversions of $\pi$ are made of adjacent elements, and denote $\pi=i_1i_2\cdots i_n$, where each $i_j$ is an inversion.
Since we are looking for an optimal sorting sequence, we may assume that all permutations in $\nonviableset$ are relevant,
which in this case means that all inversions of any permutation $\phi$ in $\nonviableset$ form a proper subset of those of $\pi$.
The reduction to {\mainproblem} in that setting is very similar to that given in the case of exchanges:
\begin{itemize}
    \item $\pi\mapsto [n]$ in the following way: $i_j\mapsto j$ for $1\leq j\leq n$;
    \item each permutation $\phi$ in $\nonviableset$ is mapped onto a subset of $[n]$ by replacing its inversions with the indices obtained in the first step;
let $\nonviableset'$ be the collection of subsets of $[n]$ obtained by applying that mapping to each $\phi$ in $\nonviableset$.
\end{itemize}
The resulting {\mainproblem} instance is then $\langle [n], \emptyset, \nonviableset', n\rangle$,
and a solution to instance $(\pi, \nonviableset, S, K)$ of {\sc guided sorting} exists if and only if a solution to instance
$\langle [n], \emptyset, \nonviableset', n\rangle$ of {\mainproblem} exists; the translation of the solution from the latter formulation to the former is straightforward.

\subsection{Main Result}\label{sct:main_result}

In the rest of this section,
we will show how to solve {\mainproblem} in time polynomial in $|\nonviableset|$ and $n$.
%\todo{A proposal for the Overview:}
% In order to provide an \emph{overview} of our algorithm, we point out that, basically, it
The algorithm mainly consists in the continuous iteration of two
%main
phases:
\begin{enumerate}
\item \emph{Double-BFS}. This phase explores the outgoing neighbourhood of the source $S$ by a breadth-first search denoted by $\texttt{BFS}_{\uparrow}$
going from lower to higher levels of $\H_n$ while avoiding the vertices in $\nonviableset$.
% Meanwhile,
$\texttt{BFS}_{\uparrow}$ collects a certain (polynomially bounded) amount
% $|\S|$
of visited vertices.
Symmetrically, the
% ingoing
incoming
neighbourhood of the target vertex $T$ is also explored by another breadth-first search $\texttt{BFS}_{\downarrow}$ going
from higher to lower levels of $\H_n$ while avoiding the vertices in $\nonviableset$,
also collecting a certain (polynomially bounded) amount
% $|\T|$
of visited vertices.
\item \emph{Compression}. If a valid solution has not yet been determined,
then a
% clever % LET THE READER DECIDE, THIS READS LIKE WE THINK VERY HIGHLY OF OURSELVES ;-)
compression technique is devised in order to shrink the size of the remaining search space.
This is possible thanks to some nice regularities of the search space
and to certain connectivity properties of hypercube graphs~\cite{GoldsLehmanRon2001, LehmanRon2001}.
% that originally arose in the context of testing monotonicity~\cite{GoldsLehmanRon2001, LehmanRon2001}.
This allows us to reduce the search space in a suitable way
and, therefore, to %retake and % ??
continue with the Double-BFS phase in order to keep the search towards valid solutions going.
\end{enumerate}

Our main contribution is summarized in the following theorem. We devote the rest of this section to
an in-depth description of the algorithms it mentions.

\begin{theorem}\label{thm:main}
Concerning the {\mainproblem}
%\todo{suggestions for catchier names are open ;-)}\
problem, the following propositions hold on any input $\langle S, T, \nonviableset, n \rangle$,
%provided that
where $d_{S,T}$ is the distance between $S$ and $T$.
\begin{enumerate}
\item There exists an algorithm for solving the \textsc{Decision-Task} of
{\mainproblem} whose time complexity is:
\[O(\min(\sqrt{|\nonviableset|\, d_{S,T}\, n}, |\nonviableset|)\, |\nonviableset|^2\, d^3_{S,T}\, n ).\]
\item There exists an algorithm for solving the \textsc{Search-Task} of
{\mainproblem} whose time complexity is:
\[O(\min(\sqrt{|\nonviableset|\, d_{S,T}\, n}, |\nonviableset|)\, |\nonviableset|^2\, d^3_{S,T}\, n + |\nonviableset|^{5/2} n^{3/2} d_{S,T} ).\]
\end{enumerate}
\end{theorem}

%The proof of correctness is offered in Appendix~\ref{subsect:correctness} and the time complexity
%is analyzed in Appendix~\ref{subsect:complexity},
% The above time bounds rely on
% Hopcroft-Karp's algorithm~\cite{HK73} for computing maximum matchings in (undirected) bipartite graphs.

%% file: Section3-Algorithm-LehmanRon.tex
\subsection{On Vertex-Disjoint Paths in Hypercube Graphs}\label{sect:VertexDisjointPaths}

The proof of Theorem~\ref{thm:main} relies on connectivity properties of hypercube graphs~\cite{GoldsLehmanRon2001}.
The next result, which proves the existence of a family of certain vertex-disjoint paths in $\H_n$
that are called \emph{Lehman-Ron paths}, will be particularly useful.

\begin{theorem}[Lehman, Ron~\cite{LehmanRon2001}]\label{thm:LehmanRon}
Given $n,m\in\N$, let $\R\subseteq\H^{(r)}_n$ and $\S\subseteq\H^{(s)}_n$ with $|\R|=|\S|=m$ and $0\leq r<s\leq n$.
Assume there exists a bijection $\varphi:\S\rightarrow\R$ such that $\varphi(S)\subset S$ for every $S\in \S$.
Then, there exist $m$ vertex-disjoint directed paths in $\H_n$ whose union contains all subsets in $\S$ and $\R$.
\end{theorem}

We call tuples $\langle \R, \S, \varphi, n\rangle$ that satisfy the hypotheses of \Cref{thm:LehmanRon} \emph{Lehman-Ron} tuples,
and we refer to the quantity $d=s-r$ as the \emph{distance} between $\R\subseteq\H^{(r)}_n$ and $\S\subseteq\H^{(s)}_n$.
\citet{LehmanRon2001} give an elementary inductive proof of \Cref{thm:LehmanRon};
also, they showed that \Cref{thm:LehmanRon} does not hold if one requires that the disjoint chains exactly correspond to the given bijection $\varphi$.
Anyway, a careful and in-depth analysis of their proof, from the algorithmic perspective, yields a polynomial time algorithm for computing all the Lehman-Ron paths.
\begin{theorem}\label{thm:algo_lehmanron}
There exists an algorithm for computing all the Lehman-Ron paths within time
$O\big( m^{5/2} n^{3/2} d \big)$ on any Lehman-Ron input $\langle \R, \S$, $\varphi, n \rangle$ with $|\R|=|\S|=m$,
where $d$ is the distance between $\R$ and $\S$ and $n$ is the size of the underlying ground set.
\end{theorem}
Now we provide all the details of the algorithm sketched above as well as a proof of the time complexity stated in \Cref{thm:algo_lehmanron},
in which Menger's vertex-connectivity theorem~\cite{Diestel2005} and Hopcroft-Karp's
algorithm~\cite{HK73} for maximum cardinality matching in \emph{undirected} bipartite graphs play a major~role.

As mentioned, our proof of \Cref{thm:main} relies on certain connectivity properties of hypercube graphs,
and in particular the existence of a family of certain vertex-disjoint paths in $\H_n$ that we call \emph{``Lehman-Ron paths''},
which is guaranteed by \Cref{thm:LehmanRon}.

Although Theorem~\ref{thm:LehmanRon} was initially proved and applied in the specific area of
testing monotonicity~\cite{GoldsLehmanRon2001}, it is of independent interest and related results
could be useful in the context of packet routing on the hypercube network.
Lehman and Ron provided an elegant inductive proof of that result~\cite{LehmanRon2001}.
In the present work, we point out that a careful analysis of their proof allows us to ``extract'' a simple recursive algorithm
for computing all Lehman-Ron paths in polynomial time. We now describe that algorithm,
whose correctness follows from the arguments used by Lehman and Ron in their original proof
of Theorem~\ref{thm:LehmanRon} (see~\cite{LehmanRon2001} for more details).
Its time complexity can be derived by taking into account Hopcroft-Karp's algorithm
for computing maximum cardinality matchings in bipartite graphs~\cite{HK73}, and is analyzed in detail at the end of this section.

The algorithm that we are going to describe is named \texttt{compute\_Lehman-Ron\_paths()}.
The intuition underlying it is simply to follow the structure of Lehman and Ron's proof of
Theorem~\ref{thm:LehmanRon} and to analyze it from the algorithmic standpoint.
% Its pseudocode is given below. % well, it won't be below as it stands, so let's provide a reference instead:
Its pseudocode is given in \Cref{ALGO:vertex-disjoint-paths}. % well, it won't be below as it stands, so let's provide a reference instead:

\begin{algorithm}[t]
\caption{computing Lehman-Ron's paths.}\label{ALGO:vertex-disjoint-paths}
\DontPrintSemicolon
\nonl \SetKwProg{Fn}{Procedure}{}{}
\Fn{$\texttt{compute\_Lehman-Ron\_paths}(\R, \S, \varphi, n)$}{
\KwIn{a Lehman-Ron tuple $\langle \R, \S, \varphi, n \rangle$.}
\KwOut{a family of $m$ vertex-disjoint directed paths $\p_1, \ldots, \p_m$ in $\H_n$
       such that $\R\cup\S \subseteq \bigcup_{i=1}^m \p_i$.}
\If{$s=r+1$}{  \label{algo:vdp:l1}
	 \Return{$\texttt{compute\_paths\_from\_bijection}(\S, \varphi, n)$}; \label{algo:vdp:l2}
}
$m\leftarrow |\S|$; \hfill\tcp{assume $|\S|=|\R|$}  \label{algo:vdp:l3}
$\QQ\leftarrow\texttt{compute\_\QQ}(\S)$;  \label{algo:vdp:l4}

$\K\leftarrow\texttt{compute\_auxiliary\_network}(\R, \QQ, \S) $; \label{algo:vdp:l5} \;
$\langle\p'_1, \p'_2, \ldots, \p'_m\rangle\leftarrow\texttt{compute\_vertex\_disjoint\_paths}(\K)$; \label{algo:vdp:l6} \;
$\langle \QQ', \varphi', \varphi'' \rangle\leftarrow
	\texttt{compute\_auxiliary\_bjcts}(\langle\p'_1, \p'_2, \ldots, \p'_m\rangle, m)$ \label{algo:vdp:l7} \;
$\langle\p''_1, \p''_2, \ldots, \p''_m\rangle\leftarrow
			\texttt{compute\_Lehman-Ron\_paths}(\R, \QQ', (\varphi')^{-1}, n)$; \label{algo:vdp:l8}\;%\tcp{\small recursion}
$\langle\p_1, \p_2, \ldots, \p_m \rangle\leftarrow
		\texttt{extend\_paths}(\langle p''_1, \p''_2, \ldots, p''_m \rangle, \QQ', \varphi'', m)$; \label{algo:vdp:l9} \;
\Return{$\langle \p_1, \p_2, \ldots, \p_m\rangle$}; \label{algo:vdp:l10}
}
\end{algorithm}

The algorithm takes as input a Lehman-Ron tuple $\langle \R, \S, \varphi, n\rangle$, and outputs a family $\p_1, \p_2, \ldots, \p_m$ of Lehman-Ron paths joining $\R$ to $\S$.
Recall that Lehman-Ron tuples satisfy the following properties:
\begin{enumerate}
    \item the families of sets $\R\subseteq \H_n^{(r)}$
and $\S\subseteq \H_n^{(s)}$ are such that $|\S|=|\R|=m$,
    \item $r, s$ and $n\in \N$ are such that $0\leq r < s \leq n$, and
    \item $\varphi:\mathcal{\S}\rightarrow\mathcal{\R}$ is a bijection
such that $\forall\ S\in \mathcal{\S}: \varphi(S)\subset S$.
\end{enumerate}

As a base case of the algorithm, if $s=r+1$ (line~\ref{algo:vdp:l1}),
then the sought family of directed paths
$\p_1, \p_2, \ldots, \p_m$ is simply a set of $m$ pairwise vertex-disjoint arcs oriented from $\S$ to $\R$, which are already given by the input bijection $\varphi$ (line~\ref{algo:vdp:l2}).

We now focus on the general case $s>r+1$. To begin with, we introduce the following proposition, which was already implicit in~\cite{LehmanRon2001}
and which is actually a straightforward consequence of Theorem~\ref{thm:LehmanRon}.

\begin{proposition}\label{prop:compute_Q}
\cite{LehmanRon2001} Given $n,m\in\N$, consider two families of sets $\R\subseteq\H^{(r)}_n$ and $\S\subseteq\H^{(s)}_n$
where $|\R|=|\S|=m$ and $0\leq r<s\leq n$.
% Let $\QQ$ be the set of the vertices in $\H^{(s-1)}_n$ that are on any directed path going from
% some vertex in $\R$ to some vertex in $\S$.
% Let $\P$ be the set of the vertices in $\H^{(r+1)}_n$ that are on any directed path going from
% some vertex in $\R$ to some vertex in $\S$.
Let $\QQ$ (resp. $\P$) be the set of vertices in $\H^{(s-1)}_n$ (resp. $\H^{(r+1)}_n$) that lie on any directed path from
some vertex in $\R$ to some vertex in $\S$.
% Let $\P$ be the set of the vertices in $\H^{(r+1)}_n$ that are on any directed path going from
% some vertex in $\R$ to some vertex in $\S$.

Then, $|\QQ|\geq m$ and $|\P|\geq m$.
\end{proposition}

\begin{figure}[!htb]
    \centering
    \begin{tikzpicture}[scale=.9,transform shape,>=stealth]
      \node[circle, draw] (s1) {$s_1$};
      \node[circle, draw, right = of s1] (s2) {$s_2$};
      \node[circle, draw, right = of s2] (s3) {$s_3$};
      \node[circle, draw, right = of s3] (s4) {$s_4$};
      \node[circle, draw, right = of s4] (s5) {$s_5$};

      \node[circle, draw, below = of s1] (q1) {$q_1$};
      \node[circle, draw, right = of q1] (q2) {$q_2$};
      \node[circle, draw, right = of q2] (q3) {$q_3$};
      \node[circle, draw, right = of q3] (q4) {$q_4$};
      \node[circle, draw, right = of q4] (q5) {$q_5$};

      \node[circle, draw, below = of q1] (r1) {$r_1$};
      \node[circle, draw, right = of r1] (r2) {$r_2$};
      \node[circle, draw, right = of r2] (r3) {$r_3$};
      \node[circle, draw, right = of r3] (r4) {$r_4$};
      \node[circle, draw, right = of r4] (r5) {$r_5$};

      \node[circle, draw, above = of s3] (trg) {$\bf t$};
      \node[circle, draw, below = of r3] (src) {$\bf s$};

      \node[left = of s1, xshift=3ex] (Slbl) {$\S$};
      \node[left = of q1, xshift=3ex] (Qlbl) {$\QQ$};
      \node[left = of r1, xshift=3ex] (Rlbl) {$\R$};

      \node[above = of Slbl, xshift=7ex, yshift=-2ex] (mlbl) {$m=5$};

    \node[circle, draw, above = of Slbl, xshift=4.7ex, yshift=2ex] (varphiL) {};
    \node[circle, draw, above = of Slbl, xshift=9.3ex, yshift=2ex] (varphiR) {};
    \draw (varphiL) edge [->, dotted, thick] node [above] {$\varphi$} (varphiR);

    \draw[dashed, ultra thin, rounded corners=15pt] (-1,2.25) rectangle (1,.75); % legend
    \draw[dashed, ultra thin, rounded corners=15pt] (-1.5,.55) rectangle (7.5,-.65); % S
    \draw[dashed, ultra thin, rounded corners=15pt] (-1.5,-1.15) rectangle (7.65,-2.40); % Q
    \draw[dashed, ultra thin, rounded corners=15pt] (-1.5,-2.8) rectangle (7.5,-4); % R

      \draw[thick] (src) edge [->] (r1);
      \draw[thick] (src) edge [->] (r2);
      \draw[thick] (src) edge [->] (r3);
      \draw[thick] (src) edge [->] (r4);
      \draw[thick] (src) edge [->] (r5);

      \draw[thick] (s1) edge [->] (trg);
      \draw[thick] (s2) edge [->] (trg);
      \draw[thick] (s3) edge [->] (trg);
      \draw[thick] (s4) edge [->] (trg);
      \draw[thick] (s5) edge [->] (trg);

      \draw[thick] (r1) edge [->] (q1); \draw[thick] (r1) edge [->] (q4);
      \draw[thick] (r2) edge [->] (q1); \draw[thick] (r2) edge [->] (q2);
      \draw[thick] (r3) edge [->] (q2); \draw[thick] (r3) edge [->] (q3);
      \draw[thick] (r4) edge [->] (q3); \draw[thick] (r4) edge [->] (q5);
      \draw[thick] (r5) edge [->] (q4); \draw[thick] (r5) edge [->] (q5);

      \draw[thick] (q1) edge [->] (s1); \draw[thick] (q1) edge [->] (s4);
      \draw[thick] (q2) edge [->] (s1); \draw[thick] (q2) edge [->] (s2);
      \draw[thick] (q3) edge [->] (s2); \draw[thick] (q3) edge [->] (s3);
      \draw[thick] (q4) edge [->] (s3); \draw[thick] (q4) edge [->] (s5);
      \draw[thick] (q5) edge [->] (s4); \draw[thick] (q5) edge [->] (s5);

      \draw (s1) edge [->, bend left=35, thick, dotted] node [above, xshift=1ex, yshift=-2ex] {\footnotesize $\varphi$} (r1);
      \draw (s2) edge [->, bend left=35, thick, dotted] node [above, xshift=1ex, yshift=-2ex] {\footnotesize $\varphi$} (r2);
      \draw (s3) edge [->, bend left=35, thick, dotted] node [above, xshift=1ex, yshift=-2ex] {\footnotesize $\varphi$} (r3);
      \draw (s4) edge [->, bend left=35, thick, dotted] node [above, xshift=1ex, yshift=-2ex] {\footnotesize $\varphi$} (r4);
      \draw (s5) edge [->, bend left=35, thick, dotted] node [above, xshift=-1ex, yshift=-2ex] {\footnotesize $\varphi$} (r5);
    \end{tikzpicture}
    \caption{The auxiliary network $\K=(V_{\K}, A_{\K})$ and bijection $\varphi$.}
\label{fig:auxiliary_network}
\end{figure}
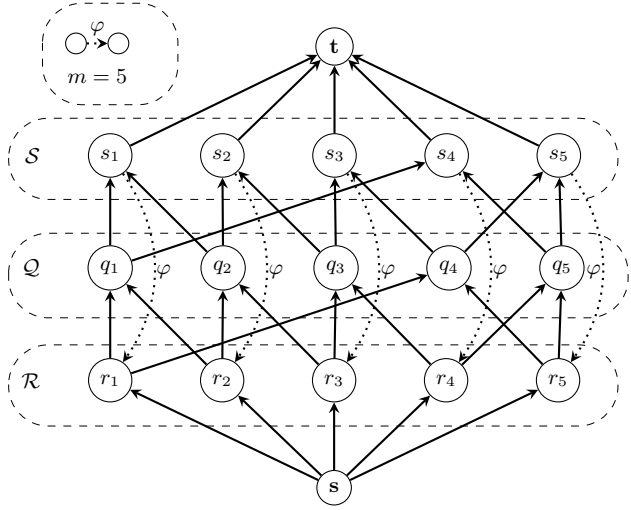

% With this in mind, the
The algorithm first computes the set $\QQ$ of all vertices in $\H^{(s-1)}_n$
that lie on any directed path from some vertex in $\R$ to some vertex in $\S$.
This step is encoded by \texttt{compute\_\QQ()} (line~\ref{algo:vdp:l4}).
The algorithm then invokes (at line~\ref{algo:vdp:l5}) a procedure called \texttt{compute\_auxiliary\_network()},
which constructs
% At this point,
a directed auxiliary network
%$\K=( V_\K, E_\K )$
$\K=( V_\K, A_\K )$ % A FOR ARCS IN DIRECTED GRAPHS
which will be useful in the following steps and is defined by:
% by invoking the procedure \texttt{compute\_auxiliary\_network()} (at line~\ref{algo:vdp:l5}).
% The vertex set $V_\K$ is the following one:
\begin{itemize}
    \item $V_{\K} = \{\textbf{s},\textbf{t}\}\cup \R\cup\QQ\cup\S$, where $\textbf{s}$ (resp. $\textbf{t}$) is an auxiliary source (resp. target) vertex, \ie $\{\textbf{s}, \textbf{t}\} \cap (\R\cup\QQ\cup\S) = \emptyset$;%\todo{is this simply a repeat of what lies before ``and''?}\todo{\textcolor{blue}{aimed to stress the fact $\textbf{s}$ and $\textbf{t}$ are novel vertices not in $\R\cup\S\cup\QQ$}};
    \item $A_{\K}$ is defined as follows:
\begin{itemize}
    \item the source vertex $\textbf{s}$ is joined to every vertex in $\R$;
    \item for each $R\in\R$ and $Q\in\QQ$, $R$ is joined to $Q$ if and only if $R\subset Q$;
    \item similarly, for each $Q\in\QQ$ and $S\in\S$, $Q$ is joined to $S$ if and only if $Q\subset S$;
    \item finally, every vertex in $\S$ is joined to $\textbf{t}$.
\end{itemize}
\end{itemize}

% An illustration of $\K$ is offered in \figref{fig:auxiliary_network}.
\figref{fig:auxiliary_network} shows an example of an auxiliary network.

We remark that, as shown in \cite{LehmanRon2001}, the following proposition holds on $\K$.

\begin{proposition}\label{prop:m_vertex_separator}
\cite{LehmanRon2001} The minimum $(\textbf{s}, \textbf{t})$-vertex-separator of $\K$ has size $m$.
\end{proposition}

As
% an aftermath
a corollary,
% of this fact,
and by applying Menger's vertex-connectivity theorem (which is recalled
% in Theorem~\ref{thm:menger})
below),
the existence of $m$ internally-vertex-disjoint directed $(\textbf{s}, \textbf{t})$-paths,
denoted $\p'_1, \p'_2, \ldots, \p'_m$, is thus
% ensured.
guaranteed.

\begin{theorem}[Menger~\cite{Diestel2005}]\label{thm:menger}
Let
%$G=( V, E )$
$G=( V, A )$
be a directed graph, and let $u$ and $v$ be nonadjacent vertices in $V$.
Then the maximum number of internally-vertex-disjoint directed $(u,v)$-paths in $G$ equals the minimum number of vertices
from $V\setminus\{u,v\}$ whose deletion destroys all directed $(u,v)$-paths in $G$.
\end{theorem}

\subsubsection{How to compute $\p'_1, \p'_2, \ldots, \p'_m$}

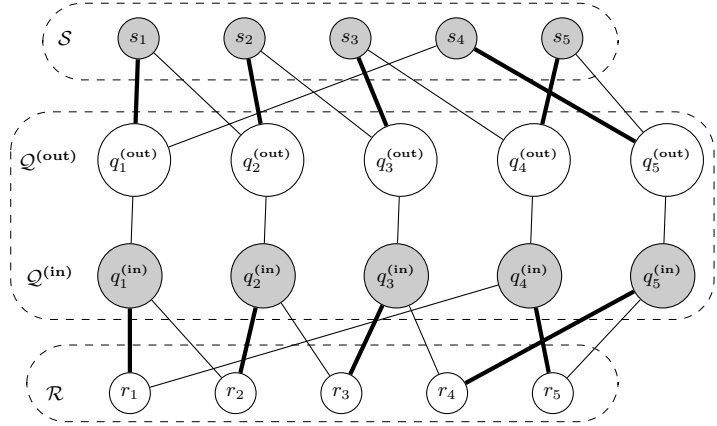
\begin{figure}[!htb]
    \centering
    \begin{tikzpicture}[scale=.85,transform shape]
      \node[circle, fill=black!20, draw] (s1) {$s_1$};
      \node[circle, fill=black!20, draw, right = of s1] (s2) {$s_2$};
      \node[circle, fill=black!20, draw, right = of s2] (s3) {$s_3$};
      \node[circle, fill=black!20, draw, right = of s3] (s4) {$s_4$};
      \node[circle, fill=black!20, draw, right = of s4] (s5) {$s_5$};

      \node[circle, draw, below = of s1, xshift=-.5ex] (q1_out) {$q^{\textbf{\tiny (out)}}_1$};
      \node[circle, draw, right = of q1_out, xshift=-.5ex] (q2_out) {$q^{\textbf{\tiny (out)}}_2$};
      \node[circle, draw, right = of q2_out, xshift=-.5ex] (q3_out) {$q^{\textbf{\tiny (out)}}_3$};
      \node[circle, draw, right = of q3_out, xshift=-.5ex] (q4_out) {$q^{\textbf{\tiny (out)}}_4$};
      \node[circle, draw, right = of q4_out, xshift=-.5ex] (q5_out) {$q^{\textbf{\tiny (out)}}_5$};

      \node[circle, draw, fill=black!20, below = of q1_out, xshift=-.5ex, yshift=2ex] (q1_in) {$q^{\textbf{\tiny (in)}}_1$};
      \node[circle, draw, fill=black!20, below = of q2_out, xshift=-.5ex, yshift=2ex] (q2_in) {$q^{\textbf{\tiny (in)}}_2$};
      \node[circle, draw, fill=black!20, below = of q3_out, xshift=-.5ex, yshift=2ex] (q3_in) {$q^{\textbf{\tiny (in)}}_3$};
      \node[circle, draw, fill=black!20, below = of q4_out, xshift=-.5ex, yshift=2ex] (q4_in) {$q^{\textbf{\tiny (in)}}_4$};
      \node[circle, draw, fill=black!20, below = of q5_out, xshift=-.5ex, yshift=2ex] (q5_in) {$q^{\textbf{\tiny (in)}}_5$};

      \node[circle, draw, below = of q1_in] (r1) {$r_1$};
      \node[circle, draw, right = of r1] (r2) {$r_2$};
      \node[circle, draw, right = of r2] (r3) {$r_3$};
      \node[circle, draw, right = of r3] (r4) {$r_4$};
      \node[circle, draw, right = of r4] (r5) {$r_5$};

      \node[left = of s1, xshift=3ex] (Slbl) {$\S$};
      \node[left = of q1_out, xshift=6ex] (Qlbl) {$\QQ^{\textbf{(out)}}$};
      \node[left = of q1_in, xshift=5.5ex] (Qlbl) {$\QQ^{\textbf{(in)}}$};
      \node[left = of r1, xshift=3ex] (Rlbl) {$\R$};

      \draw[dashed, ultra thin, rounded corners=15pt] (-1.5,.55) rectangle (7.5,-.65); % S
      \draw[dashed, ultra thin, rounded corners=15pt] (-2,-1.15) rectangle (9,-4.4); % Q
      \draw[dashed, ultra thin, rounded corners=15pt] (-1.75,-4.8) rectangle (7.5,-6); % R

      \draw[ultra thick] (r1) edge [] (q1_in); \draw[] (r1) edge [] (q4_in);
      \draw[] (r2) edge [] (q1_in); \draw[ultra thick] (r2) edge [] (q2_in);
      \draw[] (r3) edge [] (q2_in); \draw[ultra thick] (r3) edge [] (q3_in);
      \draw[] (r4) edge [] (q3_in); \draw[ultra thick] (r4) edge [] (q5_in);
      \draw[ultra thick] (r5) edge [] (q4_in); \draw[] (r5) edge [] (q5_in);

      \draw[] (q1_in) edge [] (q1_out);
      \draw[] (q2_in) edge [] (q2_out);
      \draw[] (q3_in) edge [] (q3_out);
      \draw[] (q4_in) edge [] (q4_out);
      \draw[] (q5_in) edge [] (q5_out);

      \draw[ultra thick] (q1_out) edge [] (s1); \draw[] (q1_out) edge [] (s4);
      \draw[] (q2_out) edge [] (s1); \draw[ultra thick] (q2_out) edge [] (s2);
      \draw[] (q3_out) edge [] (s2); \draw[ultra thick] (q3_out) edge [] (s3);
      \draw[] (q4_out) edge [] (s3); \draw[ultra thick] (q4_out) edge [] (s5);
      \draw[ultra thick] (q5_out) edge [] (s4); \draw[] (q5_out) edge [] (s5);

    \end{tikzpicture}
    \caption{The undirected bipartite graph $\K'=(V_{\K'}, E_{\K'})$ and a perfect matching $\M$ (thick edges).}
\label{fig:undirected_bipartite}
\end{figure}

We argue that it is possible to compute efficiently the family of directed paths $\p'_1, \p'_2, \ldots, \p'_m$ in $\K$
by finding a maximum cardinality matching in an auxiliary, \emph{undirected} bipartite graph $\K'$.
This reduction is performed by \texttt{compute\_vertex\_disjoint\_paths()} at line~\ref{algo:vdp:l6}.
% In order to perform such reduction,
% we construct an
The undirected graph $\K'=(V_{\K'}, E_{\K'})$ is obtained from the directed graph $\K$ as follows:
first, the set family $\QQ$ gets split into two (disjoint) twin set families $\QQ^{(\textbf{in})}$ and $\QQ^{(\textbf{out})}$,%\todo{unclear whether ``in'' and ``out'' mean anything or are just used to distinguish copies originating from the same element(s)},
\ie $\QQ^{(\textbf{in})}=\{ Q^{(\textbf{in})} \mid Q \in \QQ \}$ and $\QQ^{(\textbf{out})}=\{ Q^{(\textbf{out})} \mid Q \in \QQ \}$
where $\QQ^{(\textbf{in})}\cap \QQ^{(\textbf{out})} = \emptyset$ and $|\QQ^{(\textbf{in})}|=|\QQ^{(\textbf{out})}|=|\QQ|$.
Thus, the vertex set of $\K'$ is: \[V_{\K'} = \R \cup \QQ^{(\textbf{in})} \cup \QQ^{(\textbf{out})} \cup \S.\]
The
% (undirected)
edge set $E_{\K'}$ is obtained as follows:
%\todo{\textcolor{blue}{Consider any $Q\in\QQ$, each edge $e\in E_{\K}$ entering in $Q$ now becomes an edge $e'\in E_{\K'}$ entering in $Q^{(\textbf{in})}$; each edge going out of $Q$ now becomes an edge going out of $Q^{(\textbf{out})}$. Should we notice that in the text in order to motivate the notation ? Should we change the notation ?}}
\begin{itemize}
    \item for each $R\in\R$ and $Q\in\QQ$, $R$ is joined to $Q^{(\textbf{in})}$ if and only if $R\subset Q$;
\item similarly, for each $Q\in\QQ$ and $S\in\S$, $Q^{(\textbf{out})}$ is joined to $S$ if and only if $Q\subset S$;
\item finally, $Q^{(\textbf{in})}$ is joined to $Q^{(\textbf{out})}$ for every $Q\in\QQ$.
\end{itemize}

% An illustration of $\K'$ is offered in \figref{fig:undirected_bipartite}.
\figref{fig:auxiliary_network} shows an example of $\K'$.
The next proposition derives some useful properties of $\K'$.

\begin{proposition}\label{prop:K_prime}
The graph $\K'=(V_{\K'}, E_{\K'})$, as defined above, is bipartite and it admits a perfect matching.
\end{proposition}
\begin{proof}
The bipartiteness of $\K'$ follows from the bipartition $(\R\cup\QQ^{(\textbf{out})}, \QQ^{(\textbf{in})}\cup\S)$.
To see that $\K'$ admits a perfect matching,
recall that, by Proposition~\ref{prop:m_vertex_separator} and by Theorem~\ref{thm:menger},
there exist $m$ internally-vertex-disjoint directed $(\textbf{s}, \textbf{t})$-paths $\p'_1, \p'_2, \ldots, \p'_m$ in $\K$.
Then, for every $i\in [m]$, let $\p'_i = \textbf{s} R_i Q_i S_i \textbf{t}$
for some $R_i\in\R, Q_i\in\QQ, S_i\in\S$.
Finally, let us define
$\hat\QQ=\QQ\setminus \{Q\mid \exists\, i\in[m]\,\text{ s.t. } \p'_i = \textbf{s} R_i Q S_i \textbf{t} \}$.
At this point, let us consider the following matching $\M$ of $\K'$:
\begin{align*}
\M &=  \big\{ \{R_i, Q^{(\textbf{in})}_i\}, \{Q^{(\textbf{out})}_i, S_i\}
	\mid \exists\, i\in [m] \text{ s.t. } \p'_i = \textbf{s} R_i Q_i S_i \textbf{t} \big\}  \\
   & \cup \big\{\{Q^{(\textbf{in})}, Q^{(\textbf{out})}\} \mid Q\in \hat\QQ \} \big\}.
\end{align*}
Since $m=|\R|=|\S|$ and $\p'_1, \p'_2, \ldots, \p'_m$ are internally-vertex-disjoint, it follows that $\M$ is a perfect matching of $\K'$.
\end{proof}
We are in position to show how to compute $\p'_1, \p'_2, \ldots, \p'_m$ based on $\K$.
Firstly, the procedure \texttt{compute\_vertex\_disjoint\_paths()} constructs $\K'$ as explained above and computes a maximum cardinality matching
$\M$ of $\K'$ (\eg with Hopcroft-Karp's algorithm~\cite{HK73}), which is perfect
% Since $\K'$ admits a perfect matching (
by Proposition~\ref{prop:K_prime}.
% ), then $\M$ is a perfect matching.
Therefore, the following property holds: for every $Q\in \QQ$,
there exists $R\in\R$ such that $\{R, Q^{(\textbf{in})}\}\in \M$ if
and only if there exists $S\in\S$ such that $\{Q^{(\textbf{out})}, S\}\in\M$.
We can then proceed as follows:
for each $R_i\in \R$, the algorithm finds $Q_i\in\QQ$ such that $\{R_i, Q^{(\textbf{in})}_i\}\in\M$
and then it finds $S_i\in\S$ such that $\{Q^{(\textbf{out})}_i, S_i\}\in\M$.
Then, \texttt{compute\_vertex\_disjoint\_paths()}
returns the family of paths $\p'_1, \p'_2, \ldots, \p'_m$ defined as: $\p'_i = \textbf{s} R_i Q_i S_i \textbf{t}$ for every $i\in [m]$.
Since $\M$ is a perfect matching of $\K'$, the paths $\p'_1, \p'_2, \ldots, \p'_m$
are internally-vertex-disjoint.

Let $\QQ'=\{Q\mid \exists\, i\in [m]\text{ s.t. } \p'_i=\textbf{s} R_i Q S_i \textbf{t}\}$.
% Our description of the procedure \texttt{compute\_Lehman-Ron\_paths()}
% continues by noticing that, once given
Once we have computed
$\p'_1, \p'_2, \ldots$, $\p'_m$,
we can deduce two bijections that will be helpful in obtaining the wanted paths:
\[ \varphi': \R\rightarrow\QQ' \text{ and } \varphi'': \QQ'\rightarrow\S. \]
The first bijection is defined for any $R\in\R$ as $\varphi'(R)=Q$ (where $Q\in\QQ'$) provided there exists some $\p'_i$ joining
$R$ to $Q$; similarly, the second bijection is defined for any $Q\in \QQ'$ as $\varphi''(Q)=S$ (where $S\in\S$) provided
there exists some $\p'_i$ joining $Q$ to $S$.
These bijections are computed by \texttt{compute\_auxiliary\_bjcts()} at line~\ref{algo:vdp:l7}.

At this point, since the distance between $\R$ and $\QQ'$ equals~$s-1$,
a recursive call to \texttt{compute\_Lehman-Ron\_paths()} on input $\langle\R, \QQ', (\varphi')^{-1}, n \rangle$
yields, at line~\ref{algo:vdp:l8}, a family of Lehman-Ron paths $ \p''_1, \p''_2, \ldots, \p''_m$ joining $\R$ to $\QQ'$.

Indeed, we argue that it is possible to construct, starting from $\p''_1, \p''_2, \ldots, \p''_m$,
the sought family of Lehman-Ron paths $\p_1, \p_2, \ldots, \p_m$ that join $\R$ to $\S$.
Actually, this can be done just by taking into account the bijection $\varphi''$:
% In fact, notice that
since $\varphi''$ joins $\QQ'$ to $\S$,
% so
it suffices to perform the following steps in practice:
% \begin{inparaenum}
\begin{enumerate}
\item consider the last vertex $Q_i$ of $\p''_i$ (\ie the unique vertex $Q_i\in\QQ'$ such that $Q_i\in \p''_i\cap \QQ'$);
\item let $S_i=\varphi''(Q_i)$;
\item concatenate $S_i$ at the end of $\p''_i$ (\ie $\p_i= \p''_i\, S_i$).
% \end{inparaenum}
\end{enumerate}
This construction is performed by the \texttt{extend\_paths()} procedure at line~\ref{algo:vdp:l9}.
Since $\p''_1, \p''_2, \ldots, \p''_m$ are vertex-disjoint and $\varphi'':\QQ'\rightarrow\S$ is a bijection,
$\p_1, \p_2, \ldots, \p_m$ is the sought family of Lehman-Ron paths joining $\R$ to $\S$.

% This concludes the description of \texttt{compute\_Lehman-Ron\_paths()}.

%%%%%%%%%%%%%%%%%%%%%%%%% CUT POINT %%%%%%%%%%%%%%%%%%%%%%%%%%%%%%%%%%%%%%%%%%%

%of Algorithm~\ref{algo:appA} (see Appendix~\ref{app:A} for further details).

\subsubsection{Complexity Analysis (Proof of Theorem~\ref{thm:algo_lehmanron})}

We now turn to the time complexity analysis of Algorithm~\ref{ALGO:vertex-disjoint-paths}, going through each line in detail.

\begin{itemize}
    \item line~\ref{algo:vdp:l2}: \texttt{compute\_paths\_from\_bijection()} (line~\ref{algo:vdp:l2}) takes time at most $O(m)$, which corresponds to the time needed to inspect the input bijection $\varphi$.
    \item line~\ref{algo:vdp:l4}: \texttt{compute\_\QQ()} takes time at most $O(mn)$:
for each $S\in\S$, the procedure inspects the
predecessors
$N^{\texttt{in}}(S)$,
and the time bound follows from the fact that $|\S| = m$ and $|N^{\texttt{in}}(S)|\leq n$.

\item line~\ref{algo:vdp:l5}: we argue that $|V_\K|=O(mn)$ and $|A_\K|=O(m^2n)$.
Indeed, recall that $|\R|=|\S|=m$ by hypothesis;
and since every vertex of $\S$ has at most $n$ neighbours in $\QQ$,
 we have $|\QQ|\leq mn$. This in turn implies that $|V_\K|\leq 2 + 2m + mn$;
moreover, each of the $m$ vertices in $\R$ has at most $mn$ neighbours, which all lie in $\QQ$.
Therefore, $|A_\K|\leq 2m + m^2n + mn$, and the procedure
\texttt{compute\_auxiliary\_network()}
takes time at most $O(|V_\K| + |A_\K|)=O(m^2 n)$.

\item line~\ref{algo:vdp:l6}: \texttt{compute\_vertex\_disjoint\_paths()}
takes time at most $O\big( m^{5/2} n^{3/2} \big)$.
Indeed, let us consider the auxiliary (undirected) bipartite graph $\K'=(V_{\K'}, E_{K'})$ defined above.
Since $|V_\K|=O(mn)$ and $|A_\K|=O(m^2n)$,
we have $|V_{\K'}|=O(mn)$ and $|E_{\K'}|=O(m^2n)$ by construction.
A maximum cardinality matching $\M$ of $\K'$ can be computed with %the
Hopcroft-Karp's
algorithm~\cite{HK73} within time $O(\sqrt{|V_{\K'}|}\,|E_{\K'}|) = O(m^{5/2}n^{3/2})$, which yields the claimed time bound.

\item finally, lines~\ref{algo:vdp:l7} (\texttt{compute\_auxiliary\_bjcts()}) and~\ref{algo:vdp:l9} (\texttt{extend\_paths()}) take time at most $O(m)$.
\end{itemize}

To obtain the total time complexity of \texttt{compute\_Lehman-Ron\_paths()},
it is sufficient to observe that the depth of the recursion stack
(originating from line~\ref{algo:vdp:l8}) equals the distance $d= s-r$
between the families of sets that were originally given as input, $\R$ and $\S$,
and that the most expensive computation at each step of the recursion
is clearly the maximum cardinality matching computation
that is performed on the auxiliary bipartite graph $\K'$.
Therefore, we conclude that the worst-case time complexity of \texttt{compute\_Lehman-Ron\_paths()} is
$O\big( m^{5/2} n^{3/2} d \big)$.
% This actually concludes the proof Theorem~\ref{thm:algo_lehmanron}.

%% file: Section3-Algorithm-Main.tex
\subsection{A Polynomial Time Algorithm For Solving \mainproblem}\label{sect:st-dPaths}

\begin{algorithm}[t]
\caption{Solving the \texttt{\mainproblem} problem.}\label{ALGO:solve}
\DontPrintSemicolon
\nonl \SetKwProg{Fn}{Procedure}{}{}
\footnotesize
\Fn{$\texttt{solve\_\mainproblem}(S, T, \nonviableset, n)$}{
\KwIn{an instance $\langle S, T, \nonviableset, n\rangle$ of \mainproblem.}
\KwOut{a pair $\langle \texttt{YES}, \p\rangle$ where the path $\p$ is
a solution to {\mainproblem} if such a path exists, \texttt{NO} otherwise.}
$d_{S,T}\leftarrow |T|-|S|$; \label{algo:solve:l1} \hfill\tcp{let $d_{S,T}$ be the distance between $S$ and $T$}
$\S \leftarrow \{S\}$; $\ell_\uparrow\leftarrow 0$; \label{algo:solve:l2}
                \tcp{init the \emph{frontier} $\S$ and its \emph{level counter} $\ell_\uparrow$ }
$\T \leftarrow \{T\}$; $\ell_\downarrow\leftarrow 0$; \label{algo:solve:l3}
                \tcp{init the \emph{frontier} $\T$ and its \emph{level counter} $\ell_\downarrow$ }
\While{\texttt{TRUE}}{ \label{algo:solve:l4}
    $\langle \S, \T, \ell_\uparrow, \ell_\downarrow \rangle\leftarrow
        \texttt{double-bfs\_phase}(\S, \T, \nonviableset, \ell_\uparrow, \ell_\downarrow, d_{S,T}, n)$; \label{algo:solve:l5} \;
    \If{ $\S=\emptyset$ \texttt{ OR } $\T=\emptyset$ \texttt{ OR }
                ($\ell_\uparrow + \ell_\downarrow = d_{S, T}$ \texttt{ AND } $\S\cap \T=\emptyset$)}{ \label{algo:solve:l6}
    \Return{\texttt{NO}}; \label{algo:solve:l7}
    }
    \If{ $\ell_\uparrow + \ell_\downarrow = d_{S,T}$ \texttt{ AND } $\S\cap\T\neq \emptyset$ }{ \label{algo:solve:l8}
        $\p\leftarrow\texttt{reconstruct\_path}(\S, \T, n)$; \label{algo:solve:l9} \;
        \Return{$\langle \texttt{YES}, \p \rangle$}; \label{algo:solve:l10}
    }
    $\texttt{returned\_val}\leftarrow\texttt{compression\_phase}(\S, \T, \nonviableset,
                            \ell_\uparrow, \ell_\downarrow, d_{S,T}, n)$; \label{algo:solve:l11} \;
    \lIf{$\texttt{returned\_val}=\langle \texttt{YES}, \p \rangle$}{\Return{$\p$};} \label{algo:solve:l12}
    \lElse{$\T\leftarrow \texttt{returned\_val}$;} \label{algo:solve:l13}
}
}
\end{algorithm}
% This section describes
We now describe a polynomial time algorithm for solving {\mainproblem}, called
% The procedure that we are about to describe is named
\texttt{solve\_\mainproblem()}, which
% it
takes as input an instance $\langle S, T, \nonviableset, n\rangle$ of \mainproblem, and
%it aims to return
returns a pair $\langle \texttt{YES}, \p\rangle$ where $\p$ is a directed path in $\H_n$ that goes from source $S$ to target $T$
avoiding $\nonviableset$ if such a path exists (otherwise, the algorithm simply returns \texttt{NO}).
Algorithm~\ref{ALGO:solve} shows the pseudocode for that procedure.
The rationale at the base of \texttt{solve\_\mainproblem()} consists in the continuous iteration of two major phases:
% The first one is named
\texttt{double-bfs\_phase()} (line~\ref{algo:solve:l5})
% while the second one is dubbed
and
\texttt{compression\_ phase()} (line~\ref{algo:solve:l11}).
Throughout computation, both phases alternate repeatedly
% (being within the \texttt{while-loop} at line~\ref{algo:solve:l4}), one round after another,
until a final state of termination is eventually reached
(either at line~\ref{algo:solve:l7}, line~\ref{algo:solve:l10} or line~\ref{algo:solve:l12}).
At that point, the algorithm either returns a pair $\langle \texttt{YES},
\p \rangle$ where $\p$ is the sought directed path, or a negative response \texttt{NO} instead.
% In Appendix~\ref{app:omitted-proofs}, we prove that \texttt{solve\_\mainproblem()}
% always halts within polynomial time in $|\nonviableset|$ and $n$ also returning a correct answer as output.
%
% At line~\ref{algo:solve:l1}, the procedure
% starts by computing the distance $d_{S,T}=|T|-|S|$ between the source $S$ and target $T$.
% Soon after that, the computation enters within the \texttt{while-loop} at line~\ref{algo:solve:l4},
% where the \texttt{double-bfs\_phase()} is the first step that gets executed herein at line~\ref{algo:solve:l5}.
% The \texttt{double-bfs\_phase()} consists of a pair of Breadth-First Search procedures,
% that are named $\texttt{BFS}_\uparrow$ and $\texttt{BFS}_\downarrow$,
% whose pseudocode is given in Algorithm~\ref{ALGO:double_bfs}.
We now describe both phases in more detail, and give the corresponding pseudocode.

\paragraph{Breadth-First Search phases}
The first search $\texttt{BFS}_\uparrow$ starts from the source vertex $S$ and moves upward,
from lower to higher levels of $\H_n$.
Meanwhile, it collects a certain (polynomially bounded) amount of vertices that do not lie in $\nonviableset$.
In particular, at the end of any $\texttt{BFS}_\uparrow$ phase,
the number of collected vertices will always
% be strictly greater than $|\nonviableset|\, d_{S,T}$ and less than or equal to $|\nonviableset|\, d_{S,T}\,n$.
lie between $|\nonviableset|\, d_{S,T} + 1$ and $|\nonviableset|\, d_{S,T}\,n$
% Indeed, these bounds follow from line~\ref{algo:bfs:l1} of \texttt{bfs\_phase()}.
(see line~\ref{algo:bfs:l1} of \texttt{bfs\_phase()}).
The set $\S$ of vertices collected at the end of $\texttt{BFS}_\uparrow$ is called the \emph{(source) frontier} of $\texttt{BFS}_\uparrow$.
% (or \emph{source frontier})
% and it is denoted by $\S$.
All vertices within $\S$ have the same cardinality, \ie $|X_1|=|X_2|$ for every $X_1,X_2\in \S$.
Also, the procedure keeps track of the highest level of depth $\ell_\uparrow$ that is reached during $\texttt{BFS}_\uparrow$.
Thus, $\ell_\uparrow$ corresponds to the distance between the source vertex $S$ and
the current frontier $\S$,
formally, $\ell_\uparrow = |X|-|S|$ for every $X\in \S$.
Since at the beginning of the computation $\texttt{BFS}_\uparrow$ starts from the source vertex $S$,
% then
\texttt{solve\_\mainproblem()} initializes $\S$ to $\{S\}$ and $\ell_\uparrow$ to $0$ at line~\ref{algo:solve:l2}.

\begin{algorithm}[t]
\caption{Breadth-First-Search phases.}\label{ALGO:double_bfs}
\DontPrintSemicolon
\SetKwProg{Fn}{Procedure}{}{}
\footnotesize
\nonl\Fn{\texttt{double-bfs\_phase}$(\S, \T, \nonviableset, \ell_\uparrow, \ell_\downarrow, d_{S,T}, n)$}{
\setcounter{AlgoLine}{0}

$\langle \S^*, \ell^*_\uparrow \rangle\leftarrow
	\texttt{bfs\_phase}(\S, \nonviableset, \ell_\uparrow, \ell_\downarrow,
		\texttt{out}, d_{S,T}, n)$; \label{algo:two_bfs:l1} \tcp{$\text{BFS}_\uparrow$}
$\langle \T^*, \ell^*_\downarrow \rangle\leftarrow
 \texttt{bfs\_phase}(\T, \nonviableset, \ell_\downarrow, \ell^*_\uparrow,
		\texttt{in}, d_{S,T}, n)$; \label{algo:two_bfs:l2} \tcp{$\text{BFS}_\downarrow$}
\Return{$\langle \S^*, \T^*, \ell^*_\uparrow, \ell^*_\downarrow \rangle$}; \label{algo:two_bfs:l3} \;
}

\SetKwProg{SubFn}{SubProcedure}{}{}
\nonl\SubFn{\texttt{\textbf{bfs\_phase}}$(\X, \nonviableset, \ell_x, \ell_y, \texttt{drt}, d_{S,T}, n)$}{
\setcounter{AlgoLine}{0}

\While{ $1\leq|\X|\leq |\nonviableset|\, d_{S,T} \texttt{ AND } \ell_x + \ell_y < d_{S,T} $ }{ \label{algo:bfs:l1}
	$\X\leftarrow \texttt{next\_step\_bfs}(\X, \nonviableset, \texttt{drt}, n)$; \label{algo:bfs:l2} \;
	$\ell_x\leftarrow \ell_x+1$; \label{algo:bfs:l3} \;
}
\Return{$\langle \X, \ell_x \rangle$}; \label{algo:bfs:l4}
}

\setcounter{AlgoLine}{0}

\nonl\SubFn{$\texttt{\textbf{next\_step\_bfs}}(\X, \nonviableset, \texttt{drt}, n)$}{
\setcounter{AlgoLine}{0}
$\X'\leftarrow\emptyset$; \label{algo:next_bfs:l1} \;
\ForEach{$v\in \X$}{ \label{algo:next_bfs:l2}
	$\X'\leftarrow \X'\cup N^{\texttt{drt}}(v)\setminus \nonviableset$;\hfill\tcp{$N^{\texttt{drt}}$ is $N^{\texttt{in}}$ if $\texttt{drt}=\texttt{in}$, otherwise it is $N^{\texttt{out}}$ \label{algo:next_bfs:l3}}
}
\Return{$\X'$}; \label{algo:next_bfs:l4}
}
\end{algorithm}

Similarly, the second search $\texttt{BFS}_\downarrow$ starts from the target vertex $T$ and moves downward,
from higher to lower levels of $\H_n$, also collecting a certain (polynomially bounded) amount of vertices that do not lie in $\nonviableset$.
As in the previous case, this amount will always
% be strictly greater than $|\nonviableset|\, d_{S,T}$ and less than or equal to $|\nonviableset|\, d_{S,T}\, n$.
lie between $|\nonviableset|\, d_{S,T} + 1$ and $|\nonviableset|\, d_{S,T}\, n$.
The set $\T$ of vertices collected at the end of $\texttt{BFS}_\downarrow$
is called the \emph{(target) frontier} of $\texttt{BFS}_\downarrow$.
% (or \emph{target frontier})
% and it is denoted by $\T$.
All vertices within $\T$ have the same cardinality.
Also, the procedure keeps track of the \emph{lowest} level of depth $\ell_\downarrow$ that $\texttt{BFS}_\downarrow$ has reached.
Thus, $\ell_\downarrow$ corresponds to the \emph{distance} between the target vertex $T$ and the frontier $\T$,
so that $\ell_{\downarrow}=|T|-|X|$ for every $X\in\T$. Since at the beginning of the computation,
$\texttt{BFS}_\downarrow$ starts from the target vertex $T$,
\texttt{solve\_\mainproblem()} initializes $\T=\{T\}$ and $\ell_\downarrow=0$ at line~\ref{algo:solve:l3}.
\figref{fig:double_bfs_phase} provides an illustration of the behaviour of \texttt{double-bfs\_phase()}.

In summary, after any round of \texttt{double-bfs\_phase()}, we are left with two (possibly empty) frontier sets $\S$ and $\T$.
In Algorithm~\ref{ALGO:solve}, whenever $\S=\emptyset$ or $\T=\emptyset$ holds at line~\ref{algo:solve:l6},
then at least one frontier set could not proceed one level further in $\H_n$ while avoiding $\nonviableset$,
and thus the procedure halts by returning \texttt{NO} at line~\ref{algo:solve:l7}.
Similarly, whenever $\ell_\uparrow+\ell_\downarrow = d_{S,T}$ and $\S\cap\T=\emptyset$ holds at line~\ref{algo:solve:l6},
the computation halts by returning \texttt{NO} at line~\ref{algo:solve:l7} --- the underlying intuition being that
$\S$ and $\T$ have finally reached one another's level of depth without intersecting each other,
which means that $\H_n$ contains no directed path from $S$ to $T$ that avoids $\nonviableset$.

\begin{figure}[htbp]
\centering
\begin{tikzpicture}[baseline= (a).base]
\node[scale=1] (a) at (0,0){
\begin{tikzcd}[
  every arrow/.append style={-,thick},
  matrix of math nodes maybe/.append style={/tikz/cells={/tikz/nodes={/tikz/draw,/tikz/shape=circle,align=center,text width=\widthof{$x_1x_2x_3$}}}},
  row sep=normal,
  column sep=large,
  thick,
  ]
        & 1,2,3 \dar \dlar \drar  \dar[shift left=.25cm, shorten=.1cm, -stealth]\\
\textit{\textbf{1,2}} & 1,3      \dlar \drar  & \textit{\textbf{2,3}} \\
 1 \uar
        &    \textit{\textbf{2}}         \ular \urar &   \textit{\textbf{3}} \uar\\
        & \emptyset \uar \ular \urar \ular[shift left=.25cm, shorten=.25cm, -stealth]
\end{tikzcd}
};
\end{tikzpicture}
\caption{A \texttt{double\_bfs\_phase()} on $\H_3$ that starts from $S=\emptyset$ and $T=\{1,2,3\}$.
The forbidden vertices are $\nonviableset=\{\{2\}, \{3\}, \{1,2\}, \{2, 3\}\}$,
 while the edges explored by $\texttt{BFS}_{\uparrow}$ and
$\texttt{BFS}_{\downarrow}$ are $(\emptyset, \{1\})$ and $(\{1,2,3\}, \{1,3\})$ (respectively).}\label{fig:double_bfs_phase}
\end{figure}
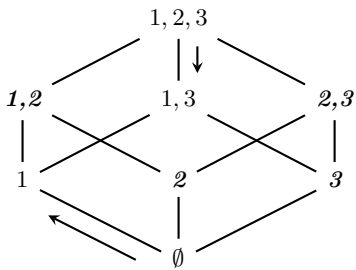

On the other
% way,
hand,
if both $\ell_\uparrow+\ell_\downarrow = d_{S,T}$ and
$\S\cap\T\neq\emptyset$ hold at line~\ref{algo:solve:l8},
% one can prove the following. For
then we can prove that for
every $S'\in\S$, there exists at least one directed path in $\H_n$
that goes from the source $S$ to $S'$ avoiding $\nonviableset$.
Similarly, for every $T'\in\T$, there exists at least one directed
path in $\H_n$ that goes from $T'$ to target $T$ avoiding $\nonviableset$.
Therefore, whenever $\S\cap\T\neq\emptyset$, the algorithm is in the right position to reconstruct
a directed path $\p$ in $\H_n$ that goes
from source $S$ to $\S\cap\T$ and from $\S\cap\T$ to target $T$ avoiding $\nonviableset$ (line~\ref{algo:solve:l9}).
In practice, the reconstruction can be implemented by maintaining a \texttt{map}
throughout the computation, which associates to every vertex $v$ (possibly visited during the BFSs)
the \emph{parent vertex}, $\texttt{parent}(v)$, which
%had
led to discover $v$ first.
As soon as $\p$ gets constructed,
% then
\texttt{solve\_\mainproblem()} returns $\langle\texttt{YES}, \p\rangle$
%is returned as output,
at line~\ref{algo:solve:l10},
and the computation halts.

\paragraph{Compression Phase}
% Still, at the conclusion of
After \texttt{double-bfs\_phase()} has completed,
the procedure \texttt{solve\_\mainproblem()} also needs to handle
the case
% in which
where $\S,\T\neq\emptyset$ and $\ell_\downarrow + \ell_\uparrow < d_{S,T}$.
The phase that starts at that point is named \texttt{compression\_phase()}
% Its pseudocode is given in
(see Algorithm~\ref{ALGO:compression_phase}).
\begin{algorithm}[t]
\caption{Compression phase.}\label{ALGO:compression_phase}
\nonl\SetKwProg{Fn}{Procedure}{}{}
\DontPrintSemicolon
\footnotesize
\Fn{$\texttt{\textbf{compression\_phase}}(\S, \T, \nonviableset, \ell_\uparrow, \ell_\downarrow, d_{S,T}, n)$}{
	$\T'\leftarrow\emptyset$; \label{algo:compression:l1} \;
	\While{\texttt{TRUE}}{ \label{algo:compression:l2}
		$\G\leftarrow\texttt{construct\_bipartite\_graph}(\S, \T, n)$; \label{algo:compression:l3}\;
		$\M\leftarrow\texttt{compute\_max\_matching}(\G, |\nonviableset|+1)$; \label{algo:compression:l4}\;
		\If{$|\M|>|\nonviableset|$}{ \label{algo:compression:l5}
    		  $\M_\S\leftarrow \{X\in\S\mid \exists\, Y\in \T\,\text{ s.t } ( X,Y ) \in \M\}$; \label{algo:compression:l6}\;
		  $\M_\T\leftarrow \{Y\in\T\mid \exists\, X\in \S\,\text{ s.t. } ( X,Y ) \in \M\}$; \label{algo:compression:l7}\;
		  $\{\p_1, \ldots, \p_{|\M|}\}\leftarrow
			\texttt{compute\_Lehman-Ron\_paths}(\M_\S, \M_\T, \M, n)$; \label{algo:compression:l8}\;
			$\p\leftarrow\texttt{reconstruct\_path}(\S, \T, \{\p_i\}_{i=1}^{|\M|}, n)$; \label{algo:compression:l9} \;
			\Return{$\langle \texttt{YES}, \p \rangle$}; \label{algo:compression:l10}
		}
		$\X\leftarrow \texttt{compute\_min\_vertex\_cover}(\G, \M)$; \label{algo:compression:l11}\;
		$\X_{\S}\leftarrow \X\cap \S$; $\X_\T\leftarrow \X\cap \T$; \label{algo:compression:l12}\;
		$\T'\leftarrow\T'\cup \X_\T$; \label{algo:compression:l13}\;
		$\langle \S, \T, \ell_\uparrow, \ell_\downarrow \rangle\leftarrow
			\texttt{double-bfs\_phase}(\X_\S, \T, \nonviableset,
				\ell_\uparrow, \ell_\downarrow, d_{S,T}, n)$; \label{algo:compression:l14} \;
		\If{ $\S=\emptyset$ \texttt{ OR } ($\ell_\downarrow + \ell_\uparrow = d_{S,T}$
						\texttt{ AND } $\S\cap \T=\emptyset$ )}{ \label{algo:compression:l15}
			\Return{$\T'$}; \label{algo:compression:l16}
		}
		\If{ $\ell_\uparrow + \ell_\downarrow = d_{S,T}$
				\texttt{ AND } $\S\cap\T\neq \emptyset$ }{ \label{algo:compression:l17}
			$\p\leftarrow\texttt{reconstruct\_path}(\S, \T, n)$; \label{algo:compression:l18} \;
			\Return{$\langle \texttt{YES}, \p \rangle$}; \label{algo:compression:l19}
		}
	}
}
\end{algorithm}
% The
%\texttt{compression\_phase()}
This procedure takes as input a tuple
$\langle \S, \T, \nonviableset, \ell_\uparrow, \ell_\downarrow, d_{S,T}, n \rangle$,
where $\S$ and $\T$ are the current frontier sets.
Recall that %the lower bound
$|\T|>|\nonviableset|\, d_{S,T}$ holds %on the size of $\T$
due to line~\ref{algo:bfs:l1} of \texttt{bfs\_phase()}.
%Moreover,
Also, $\nonviableset\subseteq\wp_n$ is the set of forbidden vertices;
$\ell_\uparrow$ is the level counter of $\S$ and $\ell_\downarrow$ is that
of $\T$; finally $d_{S,T}$ is the distance between the source $S$ and the target $T$, and $n$ is the size of the ground set.
% \textcolor{blue}{(to clarify, recall that $\langle S,T,\F,n \rangle$ is the
%  {\mainproblem}'s instance which is aimed to be solved)}.\todo{\small might be confusing: $S$ and $T$ do not appear in \texttt{compression\_phase()}}
The output returned by \texttt{compression\_phase()} is either a path $\p$ going from source $S$ to target $T$
avoiding $\nonviableset$ \texttt{or} a subset $\T'\subset \T$ such that the following two basic properties hold:
%
%\item

(1) $|\T'|\leq |\nonviableset|\, d_{S,T}$, and
%\item
(2) if $\p$ is any directed path in $\H_n$ going from
$\S$ to $\T$ avoiding $\nonviableset$, then $\p$ goes from $\S$ to $\T'$.
%\end{enumerate}

This frontier set $\T'$ is dubbed the \emph{compression} of $\T$. The underlying rationale goes as follows.
On one hand, because of (1), it is possible to keep the search going on by applying yet another round of \texttt{double-bfs\_phase()} on input $\S$ and $\T'$
(in fact, the size of $\T$ has been compressed down to $|\T'|\leq |\nonviableset|\, d_{S,T}$,
thus matching the threshold condition ``$|\X|\leq |\nonviableset|\, d_{S,T}$" checked at line~1 of \texttt{bfs\_phase()}).
On the other hand, because of (2), it is indeed sufficient to seek for a directed path in $\H_n$ that goes from $\S$ to $\T'$ avoiding $\nonviableset$,
namely, the search can actually forget about $\T\setminus \T'$ because it leads to a dead end.
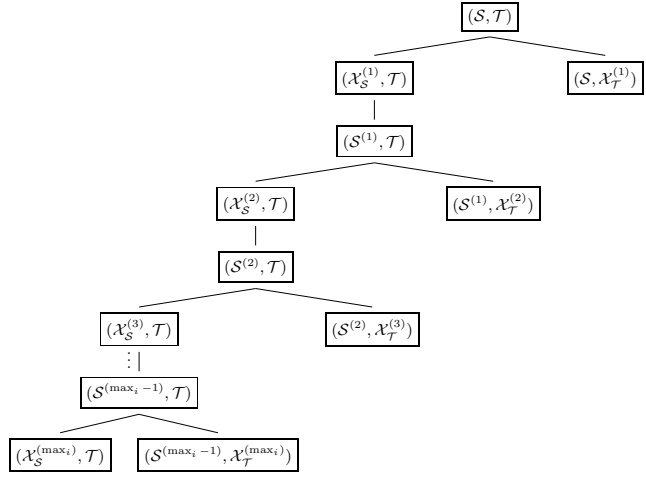
\begin{figure}[!htb]
    \centering
    \begin{tikzpicture}[scale=.68, level distance=35pt,sibling distance=9pt]
    \Tree [. \framebox{$(\S, \T)$}
        \edge node[left, xshift=-1ex]{};
    [. \framebox{$(\X^{(1)}_{\S}, \T )$}
     \edge node[left,xshift=-1ex]{};
     [. \framebox{$(\S^{(1)}, \T)$}
        \edge node[left, xshift=-1ex]{};
        [. \framebox{$(\X^{(2)}_{\S}, \T)$}
            \edge node[]{};
            [. \framebox{$(\S^{(2)}, \T)$}
                \edge node[left, xshift=-1ex]{};
                [. \framebox{$(\X^{(3)}_{\S}, \T)$}
                    \edge node[left, xshift=-.25ex, yshift=1ex]{$\vdots$};
                    %[. \framebox{$\cdots$}
                    %   \edge node[]{};
                        [. \framebox{$(\S^{(\max_i-1)}, \T)$}
                        \edge node[left, xshift=-1ex]{};
                            [. \framebox{$(\X^{(\max_i)}_{\S}, \T)$}
                            ]
                        \edge node[right, xshift=1ex]{};
                        [. \framebox{$(\S^{(\max_i-1)}, \X^{(\max_i)}_\T)$}
                        ]
                        ]
                    %]
                ]
                \edge node[right, xshift=1ex]{};
                    [. \framebox{$(\S^{(2)}, \X^{(3)}_\T)$}
                    ]
            ]
        ]
        \edge node[right, xshift=1ex]{};
        [. \framebox{$(\S^{(1)}, \X^{(2)}_\T)$}
        ]
     ]
    ]
    \edge node[right, xshift=1ex]{};
    [. \framebox{$(\S, \X^{(1)}_\T)$} ]
    ]
    \end{tikzpicture}
    \caption{The frontier sets of the \texttt{compression\_phase()}.}
\label{fig:compression_phase}
\end{figure}
% The
We now describe \texttt{compression\_phase()} in more details,
and give a graphical summary in \figref{fig:compression_phase}.
%At every iteration until a final state is reached (line~\ref{algo:compression:l10}, line~\ref{algo:compression:l16} or line~\ref{algo:compression:l19}),
%  goes as follows.
% At line~\ref{algo:compression:l1}, $\T'$ gets initialized to be the empty set. % REDUNDANT
% At line~\ref{algo:compression:l2}, the procedure enters within a \texttt{while-loop}
% that lasts until a final state of termination is eventually reached
% (either at line~\ref{algo:compression:l10}, line~\ref{algo:compression:l16} or line~\ref{algo:compression:l19}).
% Within the \texttt{while-loop} at line~\ref{algo:compression:l2}
% the procedure performs the following operations.
The procedure repeatedly builds an undirected bipartite graph
%
% Firstly, at line~\ref{algo:compression:l3} of \texttt{compression\_phase()},
% an undirected bipartite graph
$\G=(V_\G, E_\G)$,
% gets constructed in the following way:
% \begin{inparaenum}
% \item its vertex set  consists into the bipartition $(\S,\T)$;
where $V_\G=\S\cup\T$ and
% \item any two vertices $U\in\S$ and $V\in\T$ are adjacent if and only if there exists at least one directed
% path in $\H_n$ that goes from $U$ to $V$
% (\ie for
every vertex
$U\in\S$ is adjacent to a vertex $V\in\T$
% it holds $\{U, V\}\in E_\G$ if and only
if and only if $U\subset V$.
% ).
% \end{inparaenum}
%
% In short, we shall write $\G=(V_\G, E_\G)=((\S, \T,), \subset)$. % NOT USED
%
% Secondly, let us denote by $m^*$ the size of a maximum cardinality matching of $\G$.
It then uses the procedure \texttt{compute\_max\_matching()} to find a matching
% Then, at line~\ref{algo:compression:l4} of \texttt{compression\_phase()},
% a bipartite matching
$\M$ of size $|\M|=\min(m^*, |\nonviableset|+1)$,
  where $m^*$ denotes the size of a maximum cardinality matching of $\G$.
% gets computed on input $\G$.
% \todo{I'm not sure yet that these practical details have their place in the main text, but let's keep them for now (if we end up with enough space, let's keep that)}
Notice that the following holds due to line~\ref{algo:bfs:l1} of $\texttt{bfs\_phase()}$:
\[|V_\G| = |\S| + |\T| \leq 2\,|\F|\, d_{S,T}\, n,\]
thus, we have the following bound on the size of its edge set:
\[|E_\G| \leq |V_\G|^2 \leq 4\,|\F|^2\, d^2_{S,T}\, n^2. \]

\begin{algorithm}[t]
\caption{Self-Reduction for computing $\M$.}\label{ALGO:self_reduction}
\nonl\SetKwProg{Fn}{Procedure}{}{}
\DontPrintSemicolon
\footnotesize
\Fn{$\texttt{\textbf{self-reduction}}(\G,k)$}{
$\M\leftarrow \emptyset$;\label{algo:self-red:l1}\;
  \lIf{$k=0$}{
      \Return{$\M$;\label{algo:self-red:l2}}
    }
    $\hat{v}\leftarrow$ pick one vertex $\hat{v}\in V$ having maximum degree $\delta(\hat{v})$ in $\G$;\label{algo:self-red:l3}\;
    \If{$\delta(\hat{v})<k$\label{algo:self-red:l4}}{
      $\M\leftarrow$ compute a matching $\M$ of $\G$ s.t. $|\M| = \min(m^*, k)$, with the Hopcroft-Karp's algorithm~\cite{HK73};\label{algo:self-red:l5}\;
    }
    \If{$\delta(\hat{v})\geq k$\label{algo:self-red:l6}}{
      $\G'\leftarrow$ remove $\hat{v}$ from $\G$; and call the resulting graph $\G'$;\label{algo:self-red:l7}\;
      $\M'\leftarrow \texttt{self-reduction}(\G’, k-1)$;\label{algo:self-red:l8}\;
      $\M\leftarrow$ there must be at least one edge $\{u,\hat{v}\}\in E_\G$ such that $u$ is not matched in $\M’$,
      therefore, add $\{u,\hat{v}\}$ to $\M’$; and assign the resulting matching to $\M$;\label{algo:self-red:l9}\;
    }
    \Return{$\M$;\label{algo:self-red:l10}}
}
\end{algorithm}

The fact is that, given that we are content with a cardinality matching of size at most $k = |\nonviableset| + 1$,
it is worth applying the following recursive $\texttt{self-reduction}(\G,k)$ (Algorithm~\ref{ALGO:self_reduction}),
on input $(\G, |\nonviableset|+1)$, in order to shrink the upper bound on the size of $|E_\G|$ from $|V_\G|^2$ down to $|V_\G|\cdot |\nonviableset|$:
at line~\ref{algo:self-red:l1}, $\M\leftarrow\emptyset$ is initialized to the empty set. At line~\ref{algo:self-red:l2}, if $k=0$, the empty matching $\M=\emptyset$ is returned.
Then, at line~\ref{algo:self-red:l3}, let $\hat{v}\in V$ be some vertex having maximum degree $\delta(\hat{v})$ in $\G$.
If $\delta(\hat{v})<k$ at line~\ref{algo:self-red:l4}, the Hopcroft-Karp's algorithm~\cite{HK73} is invoked at line~\ref{algo:self-red:l5}
to compute a matching $\M$ of $\G$ such that $|\M| = \min(m^*, k)$, where $m^*$ is the maximum cardinality of any matching in $\G$.
In practice, this step can be implemented in the same manner as a maximum cardinality matching procedure,
\eg as Hopcroft-Karp's algorithm~\cite{HK73}, although with the following basic variation:
if the size of the augmenting matching $\M$ eventually reaches the cut-off value $k$,
then \texttt{compute\_max\_matching()}
% promptly
returns $\M$ and halts (\ie even if $m^* > k$).
Otherwise, $\delta(\hat{v})\geq k$ holds at line~\ref{algo:self-red:l6}. So, at line~\ref{algo:self-red:l7},
let $\G'$ be the graph obtained from $\G$ by removing $\hat{v}$ and all of its adjacent edges;
next, it is invoked $\texttt{self-reduction}(\G’, k-1)$ at line~\ref{algo:self-red:l8}, recursively; and, then, the returned matching is assigned to $\M'$.
Since $\delta(\hat{v})\geq k$, there must be at least one edge $\{u,\hat{v}\}\in E_\G$ such that $u$ is not matched in $\M’$,
therefore, $\{u,\hat{v}\}$ is added to $\M’$; and the corresponding matching is assigned to $\M$, at line~\ref{algo:self-red:l9}.
Finally, $\M$ is returned at line~\ref{algo:self-red:l10}. In so doing, as shown in Lemma~\ref{lemma:complexity_compression_phase},
the complexity of $\texttt{compute\_max\_matching()}$, at line~\ref{algo:compression:l4} of $\texttt{compression\_phase()}$ (Algorithm~\ref{ALGO:compression_phase}),
is going to improve by a factor $n\cdot d_{S,T}$.

%
% Subsequently, at line~\ref{algo:compression:l5}, the \texttt{compression\_phase()} tests wether the size of the
% matching $|\M|$ does exceed the number of the forbidden vertices $|\nonviableset|$.
% Then, t
The course of the next actions depends on
%  the size
$|\M|$:% of the matching:

\begin{enumerate}
\item
% \emph{Case 1.}
If $|\M| = |\nonviableset|+1$,
then the procedure relies on Theorem~\ref{thm:algo_lehmanron}
% in order
to compute a family $\p_1, \p_2, \ldots, \p_{|\M|}$ of $|\M|$
vertex-disjoint directed paths in $\H_n$ that go from $\S$ to $\T$.
In order to do that, the procedure considers the subset $\M_\S\subseteq\S$ (resp. $\M_\T\subseteq\T$) of all vertices in $\S$ (resp. in $\T$)
that are
% adjacent
incident
to some edge in $\M$ (lines~\ref{algo:compression:l6} and~\ref{algo:compression:l7}).
% The subset $\M_\T\subseteq\T$, defined analogously but with respect to $\T$,
% is also considered (line~\ref{algo:compression:l7}).
Notice that the matching $\M$ can be viewed as a bijection between $\M_\S$ and $\M_\T$.
% Whence,
Then, the algorithm underlying Theorem~\ref{thm:algo_lehmanron}
gets invoked on input $\langle \M_\S, \M_\T, \M, n \rangle$ (line~\ref{algo:compression:l8}).
Once all the Lehman-Ron paths $\p_1, \p_2, \ldots, \p_{|\M|}$ have been
% given,
found,
it is then possible to reconstruct the sought directed path $\p$ in $\H_n$ that goes from
source $S$ to target $T$  avoiding $\nonviableset$ (line~\ref{algo:compression:l9}).
In fact, since $|\M|>|\nonviableset|$ by hypothesis, and since $\p_1, \p_2, \ldots, \p_{|\M|}$
are distinct and pairwise vertex-disjoint,
there must exist at least one path $\p_i$ that goes from $\S$ to $\T$  avoiding $\nonviableset$.
It is therefore sufficient to find such a path $\p_i=v_0v_1\cdots v_k$ by direct inspection.
At that point, it is possible to reconstruct a path $\p$ going from $S$ to $v_0$ (because $v_0\in\S$),
as well as a path going from $v_k$ to $T$ (because $v_k\in\T$).
% \todo{I'm not sure yet that these practical details have their place in the main text, but let's keep them for now (if we end up with enough space, let's keep that)}
As already mentioned, in practice, the reconstruction can be implemented by maintaining a \texttt{map}
that associates to every vertex $v$ (eventually visited during the BFSs)
the parent vertex that had led to discover $v$ first.
Then, $\langle\texttt{YES},\p\rangle$ is returned at line~\ref{algo:compression:l10}.

\item
% \emph{Case 2.}
If $|\M|\leq |\nonviableset|$, then the \texttt{compression\_phase()}
aims to \emph{compress} the size of $\T$ down to $|\T'|\leq|\nonviableset|\, d_{S,T}$
% . In order to do that, it proceeds
as follows. Notice that in this case
$\M$ is a maximum cardinality matching of $\G$, because $|\M|\leq |\nonviableset|$.
So, the algorithm computes a minimum cardinality
vertex-cover $\X$ of $\G$ at line~\ref{algo:compression:l11},
whose size is $|\M|$ by K\"{o}nig's theorem~\cite{Diestel2005}.
% Notice that the size of $\X$ equals that of $\M$.
% With this in mind, at line~\ref{algo:compression:l12}, the
The algorithm then proceeds at line~\ref{algo:compression:l12} by considering the set $\X_\S=\X\cap\S$ (resp. $\X_\T=\X\cap\T$)
of all vertices that lie both in the vertex-cover $\X$ and in the frontier set $\S$ (resp. $\T$).
% moreover, the set $\X_\T=\X\cap\T$ is considered as well at line~\ref{algo:compression:l12}.
Here, it is crucial to notice that both
$|\X_\S|\leq |\nonviableset|$ and $|\X_\T|\leq |\nonviableset|$ hold,
because $|\X|=|\M|\leq|\nonviableset|$.
The fact that, since $\X$ is a vertex-cover of $\G$,
% then
any directed path in $\H_n$ that goes from $\S$ to $\T$ must go either from
$\X_\S$ to $\T$ or from $\S\setminus \X_\S$ to $\X_\T$ plays a pivotal role. Stated otherwise,
there exists no directed path in $\H_n$ that goes from $\S\setminus \X_\S$ to $\T\setminus \X_\T$,
simply because $\X$ is a vertex cover of $\G$.
At that point, the compression $\T'$ gets enriched with
$\X_\T$ at line~\ref{algo:compression:l13}.

% Soon after, the
Then, \texttt{compression\_phase()}
% continues by seeking for any
seeks a directed
path in $\H_n$ that eventually goes from $\X_S$ to $\T$.
This is done at line~\ref{algo:compression:l14} by running
% round of
\texttt{double-bfs\_phase()}
% is executed
on %input
$\langle \X_\S, \T, \nonviableset, \ell_\uparrow, \ell_\downarrow, d_{S,T}, n \rangle$.
Since $|\X_\S|\leq |\nonviableset|$, that execution results into an update of
both the frontier set $\S$ and of its level counter $\ell_\uparrow$.
Let %us denote by
$\S^{(i+1)}$ be the updated value of $\S$ and let $\ell^{(i+1)}_\uparrow$ be that of $\ell_\uparrow$.
Note that, since $|\T| > |\nonviableset|\, d_{S,T}$ holds as a pre-condition of %the
\texttt{compression\_phase()},
neither $\T$ nor $\ell_\downarrow$ are ever updated at line~\ref{algo:compression:l14}.
Upon completion of this supplementary \texttt{double-bfs\_phase()},
if $\S^{(i+1)}=\emptyset$ or both
$\ell^{(i+1)}_\uparrow+\ell_\downarrow = d_{S,T}$ \emph{and}
$\S^{(i+1)}\cap\T=\emptyset$ at line~\ref{algo:compression:l15},
then %the procedure
 $\T'$
%soon after
is returned at line~\ref{algo:compression:l16} of \texttt{compression\_phase()}.

Otherwise, if $\ell^{(i+1)}_\uparrow + \ell_\downarrow = d_{S,T}$ \emph{and} $\S^{(i+1)}\cap\T\neq\emptyset$ at line~\ref{algo:compression:l17},
the sought directed path $\p$ in $\H_n$ that goes from source $S$ to target $T$ avoiding $\nonviableset$
can be reconstructed from $\S^{(i+1)}$ and $\T$ at line~\ref{algo:compression:l18},
so that \texttt{compression\_phase()} returns $\langle \texttt{YES}, \p \rangle$
and halts soon after at line \ref{algo:compression:l19}.

% Still,
Otherwise, if $\S^{(i+1)}\neq\emptyset$ and $\ell^{(i+1)}_\uparrow + \ell_\downarrow < d_{S,T}$,
the next iteration will run on the novel frontier set $\S^{(i+1)}$ and its updated level counter $\ell^{(i+1)}_\uparrow$.
It is not difficult to prove that each iteration increases $\ell_\uparrow$ by at least one unit,
so that the \texttt{while-loop} at line~\ref{algo:compression:l2} of \texttt{compression\_phase()}
can be iterated at most $d_{S,T}$ times overall.
In particular, this fact implies that $|\T'|\leq |\nonviableset|\, d_{S,T}$
always holds at line~16 of %the
\texttt{compression\_phase()}.
\end{enumerate}

\figref{fig:compression_phase} illustrates the family of all frontier sets considered throughout \texttt{com\-pres\-sion\_phase()},
where the following notation is assumed:
% 1.
$\max_i$ is the total number of iterations of the \texttt{while-loop}
at line~\ref{algo:compression:l2} of \texttt{compression\_phase()},
% 2.
$\X^{(i)}$ is the vertex-cover computed at the $i^{\mbox{\tiny th}}$ iteration of line~\ref{algo:compression:l11},
% 3.
$\X_\S^{(i)}$ and $\X_\T^{(i)}$ are the sets computed at the $i^{\mbox{\tiny th}}$ iteration of line~\ref{algo:compression:l12}, and
% 4.
$\S^{(i)}$ is the frontier set computed at the $i^{\mbox{\tiny th}}$ iteration of line~\ref{algo:compression:l14}.
% A moment's reflection reveals that the
The compression of $\T$ (possibly returned at line~16)
is
% thus
$\T'=\bigcup_{i=1}^{\max_i} \X^{(i)}_\T$.
% This concludes the description of the $\texttt{compression\_phase()}$.

% \textbf{A Remark on Decision vs Search.}
\subsection{A Remark On Decision Versus Search}\label{sect:remark_dec_search}

% The procedure \texttt{solve\_\mainproblem()}
\Cref{ALGO:solve}
tackles
% on
the \textsc{Search-Task} of \mainproblem.
% We remark that if
If
% one is content
we merely want to answer the \textsc{Decision-Task} instead,
we can simplify the algorithm
by immediately returning \texttt{YES} if $|\M|>|\nonviableset|$ at line~\ref{algo:compression:l5} of \texttt{compression\_phase()}.
This is because in that case,
Theorem~\ref{thm:LehmanRon} guarantees the existence of a family of
$|\M|>|\nonviableset|$ vertex-disjoint paths in $\H_n$
that go from the current source frontier $\S$ to the target frontier $\T$,
which suffices to conclude that at least one of those paths avoids $\nonviableset$.
This simplification
%turns out to improve
improves the time complexity of our algorithm for solving the \textsc{Decision-Task} by a polynomial factor
over that for the \textsc{Search-Task}.

\subsection{Correctness Analysis of Algorithm~\ref{ALGO:solve}}\label{subsect:correctness}
The present subsection aims to show that the procedure \texttt{solve\_\mainproblem()} is correct.
A formal statement of that is provided in the next theorem.

\begin{theorem}\label{thm:correctness_main}
Let $\mathcal{I}=\langle S, T, \F, n\rangle$ be any instance of \mainproblem.
Given $\mathcal{I}$ as input, the procedure \texttt{solve\_\mainproblem()} halts within a finite number of steps.
Moreover, it returns as output a directed path $\p$ in $\H_n$ that goes from source $S$ to target $T$ avoiding $\F$,
provided that at least one such path exists; otherwise, the output is simply \texttt{NO}.
\end{theorem}

We are going to show a sequence of results that shall ultimately lead us to prove \Cref{thm:correctness_main}.
Hereafter, it is assumed that $\langle S, T, \F, n\rangle$ is an instance
(of \mainproblem) given as input to the \texttt{solve\_\mainproblem()} procedure. \Cref{lemma:halt_double-bfs,lemma:pre_halt_compression,lemma:halt_compression} below show that procedures \texttt{double-bfs\_phase()} and \texttt{compression\_phase()}, which are called by \texttt{solve\_\mainproblem()}, halt within a finite number of steps.

\begin{lemma}\label{lemma:halt_double-bfs}
Any invocation of \texttt{double-bfs\_phase()} halts within a finite number of steps.
In particular, the \texttt{while-loop} at line~\ref{algo:bfs:l1}
of the \texttt{bfs\_phase()} iterates at most $d_{S,T}$ times.
\end{lemma}
\begin{proof}
Consider the \texttt{while-loop} at line~\ref{algo:bfs:l1} of \texttt{bfs\_phase()}.
At each iteration of line~\ref{algo:bfs:l3}, the level counter $\ell_x$ gets incremented.
Notice that this is the only line at which $\ell_x$ may be modified,
and also notice that $\ell_y$ is never modified. Therefore, $\ell_x+\ell_y$ can only increase and not decrease.
Since the \texttt{while-loop} at line~\ref{algo:bfs:l1} of \texttt{bfs\_phase()} halts as soon as $\ell_x+\ell_y=d_{S,T}$,
the thesis follows.
\end{proof}

\begin{lemma}\label{lemma:pre_halt_compression}
Each iteration of the \texttt{while-loop} at line~\ref{algo:compression:l2} of \texttt{compression\_phase()} increases
$\ell_\uparrow + \ell_\downarrow$ by at least one unit,
either until $\ell_\uparrow + \ell_\downarrow=d_{S,T}$ or until the procedure halts by reaching either
line~\ref{algo:compression:l10}, line~\ref{algo:compression:l16} or line~\ref{algo:compression:l19}.
\end{lemma}
\begin{proof} Consider any iteration of the \texttt{while-loop}
at line~\ref{algo:compression:l2} of \texttt{compression\_} \texttt{phase()}.
Let $\G$ be the bipartite graph computed at line~\ref{algo:compression:l3},
and let $\M$ be the matching of $\G$ computed at line~\ref{algo:compression:l4}.
If $|\M| > |\F|$, then line~\ref{algo:compression:l10} gets executed, so the procedure halts
% immediately. % WELL, NOT IMMEDIATELY, WE DO STUFF BEFORE THAT; SO:
within a finite number of steps by virtue of our discussion in \Cref{sect:VertexDisjointPaths}.
Otherwise $|\M|\leq |\F|$. Recall that, since $|\M|\leq |\F|$, then $\M$ is a maximum matching of $\G$;
also recall that $\X_\S = \X\cap\S$ where $\X$ is a minimum vertex cover of $\G$ (line~\ref{algo:compression:l12}).
Since $|\X|=|\M|$, then $|\X_\S|\leq|\X|=|\M|\leq|\F|$.
Moreover, since $|\M| \leq |\F|$, \texttt{double-bfs\_phase()} gets invoked
at line~\ref{algo:compression:l14} on input $\langle \X_\S, \T, \F, \ell_\downarrow, \ell_\uparrow, d_{S,T}, n \rangle$
% This invocation
and
halts within a finite number of steps by \Cref{lemma:halt_double-bfs}.
Let us analyze its behavior with respect to $\X_\S$.
If $\X_\S=\emptyset$, then \texttt{double-bfs\_phase()} returns an empty frontier set $\S$ as output,
which leads to the termination of \texttt{compression\_phase()} at line~\ref{algo:compression:l16}.
Moreover, if $\ell_\uparrow + \ell_\downarrow=d_{S,T}$,
then
% the condition for entering the \texttt{while-loop}
% at line~\ref{algo:bfs:l1} of the \texttt{bfs\_phase()} is not satisfied;
% nevertheless, the
\texttt{compression\_phase()} halts
% soon after,
either at line~\ref{algo:compression:l16} or at line~\ref{algo:compression:l19}.
Otherwise, we must have $1 \leq |\X_\S|\leq |\F|$ and $\ell_\uparrow+\ell_\downarrow < d_{S,T}$,
in that case the condition for entering the
\texttt{while-loop} at line~\ref{algo:bfs:l1} of the \texttt{bfs\_phase()} is satisfied;
therefore, at line~\ref{algo:bfs:l3} of \texttt{bfs\_phase()}, the level counter $\ell_\uparrow$
% really
gets incremented.
This implies the thesis.
\end{proof}

\begin{lemma}\label{lemma:halt_compression}
Any invocation of \texttt{compression\_phase()} halts within a finite number of steps.
In particular, the \texttt{while-loop} at line~\ref{algo:compression:l2}
of the \texttt{compression\_phase()} iterates at most $d_{S,T}$ times.
\end{lemma}
\begin{proof}
Firstly, recall Lemma~\ref{lemma:pre_halt_compression}.
Then, notice that as soon as $\ell_\uparrow + \ell_\downarrow=d_{S,T}$
the \texttt{compression\_phase()}
then halts either at line~\ref{algo:compression:l16} (if $\S\cap\T=\emptyset$)
or at line~\ref{algo:compression:l19} (if $\S\cap\T\neq\emptyset$).
This implies that the \texttt{while-loop} at
line~\ref{algo:compression:l2} of \texttt{compression\_phase()}
iterates at most $d_{S,T}$ times.
\end{proof}

We now prove some useful properties of \texttt{compression\_phase()} and \texttt{solve\_\mainproblem()}.

\begin{lemma}\label{lemma:reconstruct_path}
The following invariant is maintained at each line of \texttt{solve\_\mainproblem()}
and at each line of \texttt{compression\_phase()}.
For every $S'\in\S$ there exists a directed path in $\H_n$ that goes from $S$ to $S'$ avoiding $\F$;
similarly, for every $T'\in\T$ there is a directed path in $\H_n$ that goes from $T'$ to $T$ avoiding $\F$.
\end{lemma}
\begin{proof}
At the beginning of the procedure $\S=\{S\}$ and $\T=\{T\}$, so the thesis holds.
At each subsequent step, the only way in which a novel vertex can be added either to $\S$ or $\T$
is by invoking the \texttt{double\_bfs\_phase()},
which preserves connectivity and avoids $\F$ by construction at line~\ref{algo:next_bfs:l3} of \texttt{next\_step\_bfs()}.
\end{proof}

\begin{lemma}\label{lemma:compression_path}
Assume that any invocation of \texttt{compression\_phase()} halts by returning $\langle\texttt{YES},\p\rangle$.
Then $\p$ is a directed path in $\H_n$ that goes from source $S$ to target $T$ avoiding $\F$.
\end{lemma}
\begin{proof}
If \texttt{compression\_phase()} returns $\p$ as output, then the last iteration of the
\texttt{while-loop} at line~\ref{algo:compression:l2} must reach
either line~\ref{algo:compression:l10} or line~\ref{algo:compression:l19}:
% . Therefore, we have two cases to analyze.
\begin{enumerate}
\item
% \emph{Case 1.}
Assume that line~\ref{algo:compression:l10} is reached at the last iteration.
Then, during that iteration, the matching $\M$ (computed at line~\ref{algo:compression:l4} on input $\G$) has size $|\M| > |\F|$.
Recall that $\G$ is a bipartite graph on bipartition $(\S, \T)$.
Let $\M_\S$ (resp. $\M_\T$ be the subset of all vertices in $\S$ (resp. $\T$) that
% are adjacent
belong
to some edge in $\M$. %, and let
% $\M_\T$ be defined analogously with respect to $\T$.
Then, by Theorem~\ref{thm:LehmanRon}, there exist $|\M|$ vertex-disjoint directed paths in $\H_n$, say
$\p_1, \p_2, \ldots, \p_{|\M|}$, whose union contains all the vertices in $\M_\S$ and $\M_\T$.
Since $|\M|>|\F|$, at least one of those paths --- say, $\p_i=v_0\cdots v_k$ --- must avoid $\F$.
By Proposition~\ref{lemma:reconstruct_path}, the procedure
\texttt{reconstruct\_path()} (invoked at line~\ref{algo:compression:l9}) is
% in the position
able
to compute a directed path $\p_{S, v_0}$ in $\H_n$ that goes from $S$ to $v_0$ avoiding $\F$
(because $v_0\in\S$, being the first step of $\p_i$), and it is also
% in the position
able
to compute a directed path $\p_{v_k,T}$
that goes from $v_k$ to $T$ avoiding $\F$ (because $v_k\in\T$, being the last step of $\p_i$).
Let $\p=\p_{S, v_0} \p_i \p_{v_k,T}$ be the directed path obtained by concatenation.
\texttt{compression\_phase()} then returns $\p$ at line~\ref{algo:compression:l10}.
% \paragraph*{}
\item
% \emph{Case 2.}
Assume that line~\ref{algo:compression:l19} is reached at the last iteration.
Then, at that iteration, the condition checked at line~\ref{algo:compression:l17} of \texttt{compression\_phase()} must be satisfied;
that is, we have $\ell_\uparrow + \ell_\downarrow=d_{S,T}$ and $\S\cap\T\neq\emptyset$.
% Since $\S\cap\T\neq\emptyset$, then we can consider some vertex $X\in\S\cap\T$ arbitrarily.
Let $X$ be an arbitrary vertex in $\S\cap\T$.
By \Cref{lemma:reconstruct_path}, there exists at least one directed path $\p_{S,X}$ in $\H_n$
that goes from $S$ to $X$ avoiding $\F$ (because $X\in\S$);
similarly, there exists at least one directed path $p_{X,T}$ in $\H_n$ that goes from $X$ to $T$ avoiding $\F$ (because $X\in\T$).
Therefore, during that iteration, the procedure \texttt{reconstruct\_path()}
(invoked at line~\ref{algo:compression:l18}) is
%really in the position
able
to compute a path $\p=\p_{S,X}\p_{X,T}$ that goes from $S$ to $X$, and then from $X$ to $T$, which is the result
% Such a path $\p$ is the one
returned by \texttt{compression\_phase()} at line~\ref{algo:compression:l19}.
\end{enumerate}
\end{proof}

The following result shows two useful properties of the frontier set returned by \texttt{compression\_phase()}, for which we will need additional notation.
Denote by $\max_i$ be the number of times that the \texttt{while-loop}
at line~\ref{algo:compression:l2} gets iterated throughout
the whole execution of the \texttt{compression\_phase()}.
% RIGHT, BUT WE'RE SAYING IT AGAIN IN THE PROOF BELOW
% Notice that $1\leq \max_i\leq d_{S,T}$ by Lemma~\ref{lemma:halt_compression}.

Also, let us introduce the following notation, for each index $i\in [\max_i]$:
\begin{itemize}
\item let $\X^{(i)}$ be the vertex cover that is computed during the $i$-th iteration of line~\ref{algo:compression:l11};
\item let $\X_\S^{(i)}$ and $\X_\T^{(i)}$ be the sets computed during the $i$-th iteration of line~\ref{algo:compression:l12};
\item let $\S^{(i)}$ be the novel frontier set that is computed during the $i$-th iteration of line~\ref{algo:compression:l14};
\end{itemize}
Moreover, we assume the notation $\S^{(0)}=\S$,
so that $\X_\S^{(i)}=\S^{(i-1)}\cap \X^{(i)}$ holds for each iteration $i\in [\max_i]$.
Notice that, since $|\T| > |\F|\, d_{S,T}$ holds by hypothesis,
then $\T$ is not modified, at line~\ref{algo:compression:l14}, by the invocation of \texttt{double-bfs\_phase()}.
Indeed, $\T$ is never modified throughout the \texttt{compression\_phase()}.
Nevertheless, a novel set $\T'\subset\T$ gets constructed and possibly returned.
% \paragraph*{}
% We are now in the position to prove (1).

\begin{proposition}\label{prop:compressed_frontier}
Assume that the procedure \texttt{compression\_phase()} is invoked on input
$\langle \S, \T, \F, \ell_\uparrow, \ell_\downarrow, d_{S,T}, n \rangle$,
where $|\T| > |\F|\, d_{S,T}$ is required to hold as a pre-condition.
Also, assume that the procedure halts at line~\ref{algo:compression:l16}, returning a novel frontier set $\T'\subset \T$.
Then, the following properties hold:
\begin{enumerate}
\item $|\T'|\leq |\F|\, d_{S,T}$;
\item if $\p$ is any directed path  in $\H_n$ that goes from $\S$ to $\T$ avoiding $\F$, then $\p$ goes from $\S$ to $\T'$.
\end{enumerate}
\end{proposition}
\begin{proof}
% Let us dwell on some preliminary facts and notation.
Firstly notice that, if an invocation of the \texttt{compression\_phase()}
halts at line~\ref{algo:compression:l16} by returning a novel frontier set $\T'\subset \T$,
this means that neither line~\ref{algo:compression:l10} nor
line~\ref{algo:compression:l19} are ever reached throughout that invocation.
In particular this implies that, at each iteration $i$ of
the \texttt{while-loop} at line~\ref{algo:compression:l2},
the maximum matching $\M^{(i)}$ (computed at line~\ref{algo:compression:l4})
has size $|\M^{(i)}|\leq |\F|$; this fact is assumed throughout the whole proof.

% I DON'T LIKE INTRODUCING NOTATION IN A PROOF, SO I MOVED IT OUT BEFORE STATING THE RESULT
% Denote by $\max_i$ be the number of times that the \texttt{while-loop}
% at line~\ref{algo:compression:l2} gets iterated throughout
% the whole execution of the \texttt{compression\_phase()}.
% Notice that $1\leq \max_i\leq d_{S,T}$ by Lemma~\ref{lemma:halt_compression}.
%
% Also, let us introduce the following notation, for each index $i\in [\max_i]$:
% \begin{itemize}
% \item let $\X^{(i)}$ be the vertex cover that is computed during the $i$-th iteration of line~\ref{algo:compression:l11};
% \item let $\X_\S^{(i)}$ and $\X_\T^{(i)}$ be the sets computed during the $i$-th iteration of line~\ref{algo:compression:l12};
% \item let $\S^{(i)}$ be the novel frontier set that is computed during the $i$-th iteration of line~\ref{algo:compression:l14};
% \end{itemize}
% Moreover, we assume the notation $\S^{(0)}=\S$,
% so that $\X_\S^{(i)}=\S^{(i-1)}\cap \X^{(i)}$ holds for each iteration $i\in [\max_i]$.
% Notice that, since $|\T| > |\F|\, d_{S,T}$ holds by hypothesis,
% then $\T$ is not modified, at line~\ref{algo:compression:l14}, by the invocation of \texttt{double-bfs\_phase()}.
% Indeed, $\T$ is never modified throughout the \texttt{compression\_phase()}.
% Nevertheless, a novel set $\T'\subset\T$ gets constructed and possibly returned.
% \paragraph*{}
% We are now in the position to prove (1).
\begin{enumerate}
    \item \emph{Proof of (1).} At each iteration $i\in[\max_i]$, the minimum vertex cover $\X^{(i)}$ has size:
\[|\X^{(i)}|=|\M^{(i)}|\leq |\F|.\]
Since $\X^{(i)}_\T=\X^{(i)}\cap\T$ at line~\ref{algo:compression:l12}, then $|\X^{(i)}_\T| \leq |\X^{(i)}| \leq |\F|$.
Moreover, recall that $\T'$ gets enriched by $\X^{(i)}$ at each iteration of line~\ref{algo:compression:l13},
so that the following holds at the termination of the \texttt{compression\_phase()}:
\[\T'=\bigcup_{i=1}^{\max_i} \X^{(i)}_\T.\] Also recall that, by Lemma~\ref{lemma:halt_compression},
the \texttt{while-loop} at line~\ref{algo:compression:l2} can be iterated at most $d_{S,T}$ times, so that $\max_i\leq d_{S,T}$.
Therefore, when \texttt{compression\_phase()} terminates, we have $|\T'|\leq |\F|\, d_{S,T}$.
%
% This concludes the proof of (1).

\item \emph{Proof of (2).} In order to prove (2),
we exhibit a number of invariants which hold for each iteration of the \texttt{while-loop} at line~\ref{algo:compression:l2} of \texttt{compression\_phase()}.
In what follows, we assume that the procedure \texttt{compression\_phase()} gets invoked on input
$\langle \S,\T,\F,\ell_\uparrow, \ell_\downarrow, d_{S,T}, n \rangle$,
and that $\S^{(0)}=\S$ holds by notational convention.

\begin{itemize}
\item[] \begin{lemma}\label{lemma:inv:1}
Let $i\in[\max_i]$ be any iteration of the \texttt{while-loop} at line~\ref{algo:compression:l2} of \texttt{compression\_phase()}.
Let $\p$ be any directed path in $\H_n$ that goes from $\S^{(i-1)}$ to $\T$.
Then $\p$ goes either from $\X^{(i)}_\S$ to $\T$ or from $\S^{(i-1)}\setminus \X^{(i)}_\S$ to $\X^{(i)}_\T$.
In other words, there exists no directed path in $\H_n$ that
goes from $\S^{(i-1)}\setminus \X^{(i)}_\S$ to $\T\setminus \X^{(i)}_\T$.
\end{lemma}
\begin{proof}
Recall that $\X^{(i)}$ is a vertex cover of the bipartite graph defined as $\G^{(i)}=((\S^{(i-1)}, \T), \subset)$,
which is constructed during the $i$-th iteration of line~\ref{algo:compression:l3} within the procedure \texttt{compression\_phase()}.
Also, $\X^{(i)}_\S = \X^{(i)}\cap \S^{(i-1)}$ and $\X^{(i)}_\T = \X^{(i)}\cap \T$,
so that the existence of any directed path in $\H_n$ going from $\S^{(i-1)}\setminus \X^{(i)}_\S$ to $\T\setminus \X^{(i)}_\T$
would imply the existence of some edge of $\G^{(i)}$ that would be uncovered by $\X^{(i)}$,
contradicting the fact that $\X^{(i)}$ is vertex cover of $\G^{(i)}$.
\end{proof}

\figref{fig:undirected_bipartite_vertex_cover} illustrates the intuition underlying Lemma~\ref{lemma:inv:1}.

\begin{figure}[!htb]
    \centering
    \begin{tikzpicture}[scale=.8,transform shape,>=stealth]
    \node[circle, fill=black!30, draw] (t1) {$t_1$};
    \node[circle, draw, scale=1.5] (t1_cover) {};
    \node[circle, fill=black!30, xshift=1ex, draw, right = of t1] (t2) {$t_2$};
    \node[circle, fill=black!30, xshift=1ex, draw, right = of t2] (t3) {$t_3$};
    \node[circle, xshift=1.75ex, draw, scale=1.5, right = of t2] (t3_cover) {};
    \node[circle, fill=black!30, xshift=1ex, draw, right = of t3] (t4) {$t_4$};
    \node[circle, xshift=1.75ex, draw, right = of t3, scale=1.5] (t4_cover) {};
    \node[circle, fill=black!30, xshift=1ex, draw, right = of t4] (t5) {$t_5$};

    \node[circle, draw, below = of t1, yshift=-10ex] (s1) {$s_1$};
    \node[circle, draw, xshift=1ex, right = of s1] (s2) {$s_2$};
    \node[circle, draw, xshift=1.72ex, right = of s1, scale=1.5] (s2_cover) {};
    \node[circle, draw, xshift=1ex, right = of s2] (s3) {$s_3$};
    \node[circle, draw, xshift=1ex, right = of s3] (s4) {$s_4$};
    \node[circle, draw, xshift=1.72ex, right = of s3, scale=1.5] (s4_cover) {};
    \node[circle, draw, xshift=1ex, right = of s4] (s5) {$s_5$};
    \node[circle, draw, xshift=1.72ex, right = of s4, scale=1.5] (s5_cover) {};

    \node[circle, draw, below = of s3, yshift=-12ex] (bot) {$\bf{0}$};
    \node[circle, draw, above = of t3, yshift=12ex] (top) {$[\bf{n}]$};

    \node[left = of t1, xshift=4ex] (Tlbl) {$\T$};
    \node[left = of s1, xshift=6ex] (Slbl) {$\S^{(i-1)}$};
    \node[left = of Tlbl, xshift=7ex, yshift=6ex] (Glbl) {$\G^{(i)}$};

    \node[above = of t2, xshift=-2ex, yshift=3.5ex] (Hn) {$\H_n$};

    \node[circle, draw, above = of bot, xshift=-1ex, yshift=-2ex] (f1) {\tiny $\bf F$};
    \node[circle, draw, above = of bot, xshift=-9ex, yshift=1ex] (f2) {\tiny $\bf F$};
    \node[circle, draw, above = of bot, xshift=8.25ex, yshift=1.5ex] (f3) {\tiny $\bf F$};

    \node[circle, draw, below = of top, xshift=-1ex, yshift=3ex] (f4) {\tiny $\bf F$};
    \node[circle, draw, below = of top, xshift=-9ex, yshift=1ex] (f5) {\tiny $\bf F$};
    \node[circle, draw, below = of top, xshift=8ex, yshift=1.5ex] (f6) {\tiny $\bf F$};

    \draw[thick] (s1) edge [] (t1); \draw[thick] (s1) edge [] (t3);
    \draw[thick] (s2) edge [] (t2); \draw[thick] (s2) edge [] (t5);
    \draw[thick] (s3) edge [] (t3); \draw[thick] (s3) edge [] (t4);
    \draw[thick] (s4) edge [] (t2); \draw[thick] (s4) edge [] (t5);
    \draw[thick] (s5) edge [] (t2); \draw[thick] (s5) edge [] (t5);

    \draw[dashed, ultra thin, rounded corners=15pt] (-2,.75) rectangle (8.2,-4); % G
    \draw[dashed, ultra thin, rounded corners=15pt] (-1.5,.65) rectangle (8,-.65); % T
    \draw[dashed, ultra thin, rounded corners=15pt] (-1.75,-2.65) rectangle (8,-3.9); % S

    \draw[dashed, thick] (bot.west) edge [] (s1.south west);
    \draw[dashed, thick] (bot.east) edge [] (s5.south east);
    \draw[dashed, thick] (top.west) edge [] (t1.north west);
    \draw[dashed, thick] (top.east) edge [] (t5.north east);

    \draw[->, dotted, thick]  plot[smooth, tension=.7] coordinates {(3.4,-6.1) (2.8,-5.5) (2.5,-5) (1.75,-4.7) (1.5,-4) (.3,-3.5)};
    \draw[->, dotted, thick]  plot[smooth, tension=.7] coordinates {(3.58,-6) (3,-5.5) (3,-5) (3.15,-4.7) (3,-4) (2.23,-3.4)};
    \draw[->, dotted, thick]  plot[smooth, tension=.7] coordinates {(3.8,-6.1) (4,-5.5) (4.2,-5) (4.2,-4.7) (3.75,-4) (3.64,-3.4)};
    \draw[->, dotted, thick]  plot[smooth, tension=.7] coordinates {(4,-6) (4.5,-5.5) (4.45,-5) (4.35,-4.7) (4.45,-4) (5.4,-3.43)};
    \draw[->, dotted, thick]  plot[smooth, tension=.7] coordinates {(3.9,-6.1) (4.7,-5.5) (4.8,-5) (5.4,-4.7) (6.15,-4) (6.8,-3.3)};

    \draw[<-, dotted, thick]  plot[smooth, tension=.7] coordinates {(3.3,2.95) (2.4,2.5) (2.1,2) (1.75,1.7) (1.25,0.8) (.26,0.25)};
    \draw[<-, dotted, thick]  plot[smooth, tension=.7] coordinates {(3.58,2.9) (3.2,2.5) (3,2) (3.15,1.7) (3,0.8) (2.17,0.25)};
    \draw[<-, dotted, thick]  plot[smooth, tension=.7] coordinates {(3.8,2.9) (3.9,2.3) (3.9,2) (4,1.7) (3.75,0.8) (3.64,.35)};
    \draw[<-, dotted, thick]  plot[smooth, tension=.7] coordinates {(4,3.1) (4.35,2.2) (4.35,2) (4.4,1.7) (4.67,.6) (5.4,0.33)};
    \draw[<-, dotted, thick]  plot[smooth, tension=.7] coordinates {(4.15,3.1) (4.5,2.5) (5.25,2) (5.5,1.7) (6.5,.6) (7.35,0.35)};
    \end{tikzpicture}
    \caption{The undirected bipartite graph $\G^{(i)}=((\S^{(i-1)},\T), \subset)$, 
	and vertex cover $\X^{(i)}=(X^{(i)}_\S, X^{(i)}_\T)$ (doubly-circular nodes).}
\label{fig:undirected_bipartite_vertex_cover}
\end{figure}
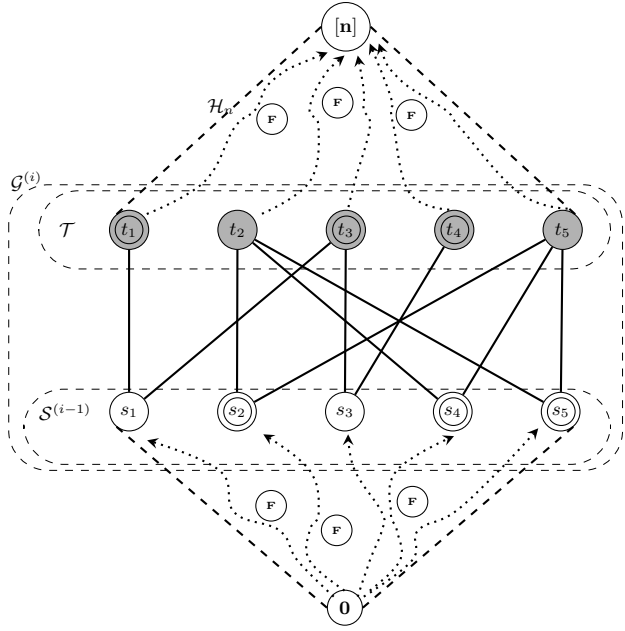

\item[] \begin{lemma}\label{lemma:inv:2}
Let $i\in[\max_i]$ be any iteration of the \texttt{while-loop} at line~\ref{algo:compression:l2} of \texttt{compression\_phase()}.
Let $U$ be any subset of $\S^{(i-1)}$ and let $V$ be any subset of $\T$. Let $\p$ be any directed path in $\H_n$ that goes from $U$ to $V$.
Then $\p$ goes from $\S$ to $V$ in $\H_n$.
\end{lemma}
\begin{proof}
Induction on $i\in[\max_i]$.
\begin{itemize}
\item \emph{Base Case.} If $i=1$, recall that $\S^{(0)}=\S$. Then $U\subseteq \S$, which  implies the base case.
\item \emph{Inductive Step.}
Let us assume, by induction hypothesis, that the claim holds for some $i\in[\max_i-1]$
and let us prove it for $i+1$.
So, let $U\subseteq \S^{(i)}$, and let $\p$ by any directed path
in $\H_n$ that goes from $U$ to $V$.
Recall that $\S^{(i)}$ is the frontier set that
is returned by an invocation of \texttt{double-bfs\_phase()}
on input $\X^{(i)}_\S$, at the $i$-th iteration of line~\ref{algo:compression:l14}, within \texttt{compression\_phase()}.
This amounts to saying that all vertices in $\S^{(i)}$ have been discovered by a BFS starting from $\X^{(i)}_\S$.
Recall that $\X_\S^{(i)}=\X^{(i)}\cap\S^{(i-1)}$ so that $\X_\S^{(i)}\subseteq \S^{(i-1)}$.
Therefore, $\p$ is indeed a directed path in $\H_n$ that goes from $\S^{(i-1)}$ to $V$ in $\H_n$.
By induction hypothesis, the thesis follows.
\end{itemize}
\end{proof}

\item[] \begin{lemma}\label{lemma:inv:3}
Let $i\in[\max_i]$ be any index of iteration of the \texttt{while-loop}
at line~\ref{algo:compression:l2} of \texttt{compression\_phase()}.
Let $\p$ be a directed path in $\H_n$ that goes from $\S$ to $\T$  avoiding $\F$.
Then, $\p$ goes either from $\X^{(i)}_\S$ to $\T$ or from $\S$ to $\bigcup_{j=1}^{i} \X^{(j)}_\T$.
\end{lemma}
\begin{proof}
Induction on $i\in [\max_i]$.

\begin{itemize}
\item \emph{Base Case.} If $i=1$, recall that $\S^{(0)}=\S$.
Then, by \Cref{lemma:inv:1}, we have that $\p$ either goes from
$X^{(1)}_\S$ to $\T$ or from $\S\setminus \X^{(1)}_\S$ to $\X^{(1)}_\T$.
If $\p$ goes from $\S\setminus \X^{(1)}_\S$ to $\X^{(1)}_\T$,
then clearly $\p$ goes from $\S$ to $\X^{(1)}_\T$.
This implies the base case.

\item \emph{Inductive Step.} Let us assume, by induction hypothesis,
that the claim holds for some $i\in[\max_i-1]$, and let us prove it for $i+1$.
By induction hypothesis, $\p$ either goes from $\X^{(i)}_\S$ to $\T$
or from $\S$ to $\bigcup_{j=1}^{i} \X^{(j)}_\T$ in $\H_n$.

If $\p$ goes from $\X^{(i)}_\S$ to $\T$  avoiding $\F$ in $\H_n$,
then $\p$ must go from $\S^{(i)}$ to $\T$:
in fact, recall that $\S^{(i)}$ is the frontier set that is returned
by the invocation of \texttt{double-bfs\_phase()} on input $\X^{(i)}_\S$,
at line~\ref{algo:compression:l14} of the \texttt{compression\_phase()}.

If $\p$ goes from $\S^{(i)}$ to $\T$ then, by Lemma~\ref{lemma:inv:1}, we also have that
$\p$ goes either from $\X^{(i+1)}_\S$ to $\T$ or from
$\S^{(i)}\setminus \X^{(i+1)}_\S$ to $\X^{(i+1)}_\T$ in $\H_n$.

If $\p$ goes from $\S^{(i)}\setminus \X^{(i+1)}_\S$ to $\X^{(i+1)}_\T$,
then $\p$ goes from $\S$ to $\X^{(i+1)}_\T$ by Lemma~\ref{lemma:inv:2}.

% In summary, we have shown that
Since $\p$ either goes from $\X^{(i+1)}_\S$ to $\T$,
or from $\S$ to $\X^{(i+1)}_\T$, or from $\S$ to $\bigcup_{j=1}^{i} \X^{(j)}_\T$ in $\H_n$,
we have that $\p$ either goes from $\X^{(i+1)}_\S$ to $\T$, or from $\S$ to
$\bigcup_{j=1}^{i+1} \X^{(j)}_\T$ in $\H_n$,
thus concluding the induction and the proof of Lemma~\ref{lemma:inv:3}.
\end{itemize}
\end{proof}
\end{itemize}

We now have everything we need to prove (2).
Let $i=\max_i$ be the last iteration of the \texttt{while-loop} at line~\ref{algo:compression:l2} of \texttt{compression\_phase()}.
Moreover, assume that $\p$ is a directed path in $\H_n$ that goes from $\S$ to $\T$ avoiding $\F$.

By \Cref{lemma:inv:3}, $\p$ either goes from
$\X^{(\max_i)}_\S$ to $\T$ or from $\S$ to $\bigcup_{i=1}^{\max_i} \X^{(i)}_\T$.
We argue that $\p$ cannot go from $\X^{(\max_i)}_\S$ to $\T$ in $\H_n$.
In fact, any such path must first visit $\S^{(\max_i)}$ in order to reach $\T$.
Then, it is sufficient to show that there exists no path that goes from $\S^{(\max_i)}$ to $\T$.
Recall that $\max_i$ is the last iteration of the \texttt{while-loop} at line~\ref{algo:compression:l2},
and by hypothesis the \texttt{compression\_phase()} halts by returning $\T'$ at line~\ref{algo:compression:l16}.
Therefore, at line~\ref{algo:compression:l15}, it must hold that $\S^{(\max_i)}=\emptyset$ or that both
$\ell_\downarrow^{(\max_i)}+\ell_\uparrow=d_{S,T}$ and $\S^{(\max_i)}\cap\T=\emptyset$.
Thus, there exists no directed path in $\H_n$ that goes from $\S^{(\max_i)}$ to $\T$.

Since $\p$ does not go from $\X^{(\max_i)}_\S$ to $\T$,
it must go from $\S$ to $\bigcup_{i=1}^{\max_i} \X^{(i)}_\T$ instead; and since
% Recall
$\T'= \bigcup_{i=1}^{\max_i} \X^{(i)}_\T$,
$\p$ must therefore go from $\S$ to $\T'$, which concludes the  proof of (2).
\end{enumerate}
\end{proof}

Now that we have established the correctness of the procedures it uses, we go back to establishing the correctness of \texttt{solve\_\mainproblem()}.

\begin{lemma}\label{lemma:pre_solve_halts_finite}
Each iteration of the \texttt{while-loop} at line~\ref{algo:solve:l4} of \texttt{solve\_\mainproblem()}
increases $\ell_\uparrow + \ell_\downarrow$ by at least one unit;
until $\ell_\uparrow + \ell_\downarrow=d_{S,T}$ or until the procedure halts by reaching
either line~\ref{algo:solve:l7}, line~\ref{algo:solve:l10} or line~\ref{algo:solve:l12}.
\end{lemma}
\begin{proof} Induction on the index $i$
of iteration of the \texttt{while-loop} at line~\ref{algo:solve:l4}.
\begin{itemize}
\item \emph{Base Case.} Consider the first iteration of the \texttt{while-loop} at line~\ref{algo:solve:l4}.
We have $\S=\{S\}$, $\T=\{T\}$, and $\ell_\uparrow=\ell_\downarrow=0$.
Therefore, if $d_{S,T}=0$, then the procedure halts immediately,
either at line~\ref{algo:solve:l7} (if $S\neq T$) or at line~\ref{algo:solve:l10} (if $S=T$).
If $d_{S,T}>0$, then a first execution of \texttt{double-bfs\_phase()}
is invoked at line~\ref{algo:solve:l5}, which halts after a finite number of steps by \Cref{lemma:halt_double-bfs}.
Notice that the condition for entering  the \texttt{while-loop}
at line~\ref{algo:bfs:l1} of \texttt{bfs\_phase()} is satisfied, so
$\ell_\uparrow+\ell_\downarrow$ gets incremented at line~\ref{algo:bfs:l3} of \texttt{bfs\_phase()}.

\item \emph{Inductive Step.}
Assume that at the $i$-th iteration of the \texttt{while-loop} at line~\ref{algo:solve:l4},
we have $\ell_\uparrow + \ell_\downarrow < d_{S,T}$. Furthermore, assume that none of the conditions
checked by \texttt{solve\_{\mainproblem}()} at line~\ref{algo:solve:l6},
line~\ref{algo:solve:l8} and line~\ref{algo:solve:l12} are satisfied.
Then, the procedure does not halt at $i$-th iteration.
Recall that \texttt{double-bfs\_phase()}, which is invoked at line~\ref{algo:solve:l5},
halts within finite time by Lemma~\ref{lemma:halt_double-bfs};
also, recall that \texttt{compression\_phase()}, which is invoked at line~\ref{algo:solve:l11},
halts within finite time by \Cref{lemma:halt_compression}.
Thus, at the end of the $i$-th iteration, line~\ref{algo:solve:l13} gets finally executed.
At line~\ref{algo:solve:l13}, the current frontier $\T$ gets replaced by the value $\T'$,
previously returned by \texttt{compression\_phase()} at line~\ref{algo:solve:l11}.
Notice that $|\T'| \leq |\F|\, d_{S,T}$ holds by \Cref{prop:compressed_frontier}.
The $(i+1)$-th iteration of the \texttt{while-loop} at line~\ref{algo:solve:l4} starts at this point.
Then, at line~\ref{algo:solve:l5}, another round of \texttt{double-bfs\_phase()} is executed.
If $\T\neq\emptyset$ and $\ell_\uparrow+\ell_\downarrow<d_{S,T}$,
the condition for entering the \texttt{while-loop} at line~\ref{algo:bfs:l1} of \texttt{bfs\_phase()} is satisfied,
so that $\ell_\uparrow+\ell_\downarrow$ gets incremented at line~\ref{algo:bfs:l3}.
If $\T=\emptyset$ or $\ell_{\uparrow}+\ell_\downarrow=d_{S,T}$,
then the procedure halts at line~\ref{algo:solve:l7}.
This implies that the invariant is maintained for each iteration $i$.
\end{itemize}
\end{proof}

\begin{proposition}\label{prop:solve_halts_finite}
The procedure \texttt{solve\_\mainproblem()} halts within a finite number of steps.
In particular, the \texttt{while-loop} at line~\ref{algo:solve:l4} iterates at most $d_{S,T}$ times.
\end{proposition}
\begin{proof}
Recall the statement of \Cref{lemma:pre_solve_halts_finite}. As soon as $\ell_\uparrow + \ell_\downarrow=d_{S,T}$,
then \texttt{solve\_\mainproblem()} halts either at line~\ref{algo:solve:l7}
(if $\S\cap\T=\emptyset$) or at line~\ref{algo:solve:l10} (if $\S\cap\T\neq\emptyset$).
In particular, this implies that the \texttt{while-loop} at line~\ref{algo:solve:l4} of
the \texttt{solve\_\mainproblem()} can be iterated at most $d_{S,T}$ times.
\end{proof}

\begin{proposition}\label{prop:correctness_yes}
Assume that \texttt{solve\_\mainproblem()} halts by returning the pair $\langle \texttt{YES}, \p\rangle$.
Then $\p$ is a directed path in $\H_n$ that goes from $S$ to $T$ avoiding $\F$.
\end{proposition}
\begin{proof}
Observe that \texttt{solve\_\mainproblem()} can return $\langle\texttt{YES}, \p\rangle$
as output only at line~\ref{algo:solve:l10} or at line~\ref{algo:solve:l12}.
In the latter case, $\p$ gets constructed at line~\ref{algo:solve:l11} by invoking \texttt{compression\_phase()},
so the thesis follows by \Cref{lemma:compression_path}.
Otherwise, assume that $\p$ is returned at line~\ref{algo:solve:l10}.
Therefore, at the last iteration of line~\ref{algo:solve:l8}, it must hold that $\S\cap\T\neq\emptyset$.
Then, let $X\in\S\cap\T$.
By \Cref{lemma:reconstruct_path} there exists a directed path $p_{S,X}$ in $\H_n$ that goes from
$S$ to $X$ avoiding $\F$ (because $X\in\S$),
and there exists another directed path $\p_{X,T}$ in $\H_n$ that goes from $X$ to $T$ avoiding $\F$ (because $X\in\T$).
Therefore, \texttt{reconstruct\_path()} at line~\ref{algo:solve:l9},
is able to compute a directed path $\p=\p_{S,X}\p_{X,T}$ in $\H_n$ that goes from $S$ to $T$
avoiding $\F$, which gets returned at line~\ref{algo:solve:l12}.
\end{proof}

\begin{lemma}\label{lemma:no_invariant}
The following invariant is maintained at each line of \texttt{solve\_\mainproblem()}.
If $\p$ is any directed path in $\H_n$ that goes from $S$ to $T$ avoiding $\F$ ,
then $\p$ goes from $\S$ to $\T$.
\end{lemma}
\begin{proof}
Induction on the index $i$ of iteration of the \texttt{while-loop} at line~\ref{algo:solve:l2}.
\begin{itemize}
\item \emph{Base Case.} Before entering the first iteration,
since $\S=\{S\}$ and $\T=\{T\}$, the thesis holds.
\item \emph{Inductive Step.} Assume that the thesis holds at the end of the $i$-th iteration.
So, let $\S^{(i)}$ and $\T^{(i)}$ be the frontier sets at the end of the $i$-th iteration.
When $i=0$, just recall that $\S^{(0)}=\{S\}$ and $\T^{(0)}=\{T\}$.
Now, at the beginning of the $(i+1)$-th iteration,
in particular at line~\ref{algo:solve:l5} of \texttt{solve\_\mainproblem()},
let $\S$ and $\T$ be the frontier sets returned by the invocation of \texttt{double-bfs\_phase()}.
If $\p$ is any directed path in $\H_n$ that goes from $S$ to $T$ avoiding $\F$,
then $\p$ goes from $\S^{(i)}$ to $\T^{(i)}$ by induction hypothesis.
It is not difficult to see that if $\p$ goes from $\S^{(i)}$ to $\T^{(i)}$ avoiding $\F$,
then $\p$ must go from $\S$ to $\T$ as well:
at this point, the reader can check that this is
a direct consequence of \texttt{double-bfs\_phase()}'s construction.
If the $(i+1)$-th iteration doesn't halt,
then the \texttt{compression\_phase()} at line~\ref{algo:solve:l11} gets invoked.
Then, let $\T'$ be the value returned by \texttt{compression\_phase()} at line~\ref{algo:solve:l11}.
By \Cref{prop:compressed_frontier},
if $\p$ is a directed path in $\H_n$ that goes from $\S$ to $\T$ avoiding $\F$, then $\p$ goes from $\S$ to $\T'$.
Thus, it is indeed correct to update $\T$ by $\T'$ at line~\ref{algo:solve:l13} of \texttt{solve\_\mainproblem()}.
This implies that the thesis holds for each iteration of the \texttt{while-loop} at line~\ref{algo:solve:l2},
until termination.
\end{itemize}
\end{proof}

\begin{proposition}\label{prop:correctness_no}
Assume that \texttt{solve\_\mainproblem()} halts by returning \texttt{NO}.
Then there is no directed path in $\H_n$ that goes from $S$ to $T$ avoiding $\F$.
\end{proposition}
\begin{proof}
Since \texttt{solve\_\mainproblem()} returns \texttt{NO},
the condition checked at line~\ref{algo:solve:l6} must be satisfied:
if $\S=\emptyset$ or $\T=\emptyset$, then there exists no directed path in $\H_n$ that goes from $\S$ to $\T$;
similarly, if $\ell_\uparrow+\ell_\downarrow=d_{S,T}$ and $\S\cap\T=\emptyset$,
then there exists no directed path in $\H_n$ that goes from $\S$ to $\T$.
By Lemma~\ref{lemma:no_invariant},
there exists no directed path in $\H_n$ that goes from $S$ to $T$ avoiding $\F$.
\end{proof}

\noindent Theorem~\ref{thm:correctness_main} follows, at this point,
from \Cref{prop:solve_halts_finite,prop:correctness_yes,prop:correctness_no}.

\subsection{Complexity Analysis}\label{subsect:complexity}

We now analyze the time complexity of \texttt{solve\_\mainproblem()}, starting with that of the procedures it relies on.

\begin{lemma}\label{lemma:complexity_double-bfs_phase}
The $\texttt{double-bfs\_phase()}$ always halts within $O(|\F|\, d^2_{S,T}\, n)$ time.
\end{lemma}
\begin{proof}
It is sufficient to prove that \texttt{bfs\_phase()} always halts within $O(|\F|\, d^2_{S,T}\, n)$ time.
Recall that, by Lemma~\ref{lemma:halt_double-bfs}, the \texttt{while-loop} at line~\ref{algo:bfs:l1} of \texttt{bfs\_phase()}
iterates at most $d_{S,T}$ times. At each iteration,
\texttt{next\_step\_bfs()} gets invoked on some input set
%\textcolor{blue}{
$\X\in\wp_n$
and flag variable $\texttt{drt}\in\{\texttt{in}, \texttt{out}\}$ (see line~\ref{algo:bfs:l2} of \texttt{bfs\_phase()}).
%}.

We argue that each of these invocations takes at most $O(|\F|\, d_{S,T}\, n)$ time.
%\textcolor{blue}{
Assume that $N^{\texttt{drt}}$ is $N^{\texttt{in}}$ when $\texttt{drt}=\texttt{in}$,
and that it is $N^{\texttt{out}}$ otherwise.
%}
Then, each invocation of \texttt{next\_step\_bfs()} takes $O(|\X|\, \max_{v\in \X}\{|N^{\texttt{drt}}(v)|\})$
time,
%\todo{warning: here and in the main text, I didn't see a definition for $\texttt{drt}$ or $N^{\texttt{drt}}(v)$}
because it involves visiting $N^{\texttt{drt}}(v)$ for each $v\in\X$;
still, in order to enter the \texttt{while-loop} at line~\ref{algo:bfs:l1} of \texttt{bfs\_phase()},
we must have $|\X|\leq |\F|\, d_{S,T}$,
and moreover we have $|N^{\texttt{drt}}(v)|=O(n)$ for every $v\in\X$.
Since the total number of iterations is bounded above by $d_{S,T}$, the bound follows.
\end{proof}

\begin{lemma}\label{lemma:complexity_compression_phase}
Assume that \texttt{compression\_phase()} gets invoked at line~\ref{algo:solve:l11} of
the procedure \texttt{solve\_\mainproblem()}. If $\texttt{compression\_phase()}$ halts without ever executing
the procedure \texttt{compute\_Lehman-Ron\_paths()} at line~\ref{algo:compression:l8},
then it halts within the following time bound:
\begin{equation}\label{eqn:compression-phase-first-time-bound}
O\left(\min\Big(\sqrt{|\F|\, d_{S,T}\, n}, |\F|\Big) |\F|^2\, d^2_{S,T}\, n \right)
\end{equation}
Otherwise, if $\texttt{compression\_phase()}$ executes
\texttt{compute\_Lehman-Ron\_paths()} at line~\ref{algo:compression:l8},
then it halts within the following time bound:
\begin{equation}\label{eqn:compression-phase-second-time-bound}
O\left(\min\Big(\sqrt{|\F|\, d_{S,T}\, n}, |\F|\Big) |\F|^2\, d^2_{S,T}\, n +
 |\nonviableset|^{5/2} n^{3/2} d_{S,T} \right)
\end{equation}
\end{lemma}
\begin{proof}
We start with some preliminary observations that will be useful in proving time
bounds~\eqref{eqn:compression-phase-first-time-bound} and \eqref{eqn:compression-phase-second-time-bound}.
Let us assume that \texttt{compression\_phase()} is
invoked on the following input $\langle \S, \T, \F, \ell_\uparrow, \ell_\downarrow, d_{S,T}, n \rangle$
at line~\ref{algo:solve:l11} of \texttt{solve\_\mainproblem()}.
We argue that the following bounds hold on the size of $\S$ and $\T$:
\begin{equation}\label{eqn:bounds-S-T}
|\S|\leq |\F|\, d_{S,T}\, n \;\;\text{ and }\;\; |\T|\leq |\F|\, d_{S,T}\, n.
\end{equation}
In fact, notice that $\S$ and $\T$ were computed during a previous invocation of $\texttt{double-bfs\_phase()}$,
at line~\ref{algo:solve:l5} of \texttt{solve\_\mainproblem()}.
Therefore, it suffices to consider the set $\X$ which is computed by passing through the
\texttt{while-loop} at line~\ref{algo:bfs:l1} of \texttt{bfs\_phase()}.
The condition for entering that \texttt{while-loop} requires $|\X|\leq |\F|\, d_{S,T}$.
Therefore, as soon as \texttt{bfs\_phase()} exits that \texttt{while-loop}, we must have $|\X|\leq |\F|\, d_{S,T}\, n$.
This implies the bounds specified by \eqref{eqn:bounds-S-T}.

Consider the bipartite graph $\G=(V_\G, E_\G)=((\S, \T), \subset)$, which is constructed
at line~\ref{algo:compression:l3} of \texttt{compression\_phase()}.
% We have the following bound on the size of its vertex set:
Since we have:
\[|V_\G| = |\S| + |\T| \leq 2\,|\F|\, d_{S,T}\, n,\]
we also have the following bound on the size of its edge set:
\[|E_\G| \leq |V_\G|^2 \leq 4\,|\F|^2\, d^2_{S,T}\, n^2. \]

We can now proceed with the proof of the two time bounds.

\begin{enumerate}
    \item In the case where  \texttt{compute\_Lehman-Ron\_paths()} never gets executed,
    recall that, at line~\ref{algo:compression:l4}, the \texttt{compression\_phase()}
computes a matching $\M$ of $\G$ such that $|\M|=\min(m^*, |\F|+1)$, where $m^*$ is the
size of a maximum cardinality matching of $\G$.
At this point, the $\texttt{self-reduction}(\G, |\F|+1)$ (Algorithm~\ref{ALGO:self_reduction}),
allows us to shrink the upper bound on the size of $|E_\G|$ from $|V_\G|^2$ down to:
\[|V_\G|\cdot |\nonviableset| \leq 2\,|\nonviableset|^2\, d_{S,T}\, n.\]
%can be implemented in order to
The total overhead introduced by $\texttt{self-reduction()}$ is only $O(|V_\G| + |E_\G|)$,
because there are at most $|V_\G|$ recursive calls, each one inspecting the neighbourhood of some node of $\G$.
So, $\M$ is computed within the following time bound $t_\M$:
\begin{align*}
t_\M & = O\left(\min(\sqrt{|V_\G|}, |\F|)\,|E_\G|\right)  \\
     & = O\left(\min\Big(\sqrt{|\F|\, d_{S,T}\, n}, |\F|\Big) |\F|^2\, d_{S,T}\, n \right)
\end{align*}
At this point, let us observe that the time complexity of \texttt{compute\_min\_vertex\_cover()},
which is invoked at line~\ref{algo:compression:l11} of \texttt{compression\_phase()},
is bounded above by the time complexity of computing $\M$ at line~\ref{algo:compression:l4}.
Also, by Lemma~\ref{lemma:complexity_double-bfs_phase},
the time complexity of the \texttt{double-bfs\_phase()}, which is invoked at line~\ref{algo:compression:l14}
of \texttt{compression\_} \texttt{phase()}, is bounded above by the same quantity.
% \paragraph*{}

% Let us assume that the \texttt{compression\_phase()} halts without ever executing the procedure
If \texttt{compute\_Lehman-Ron\_paths()} never gets executed at line~\ref{algo:compression:l8},
then during each iteration of the \texttt{while-loop}
at line~\ref{algo:compression:l2} of \texttt{compression\_phase()},
the most expensive task is that of computing the matching $\M$ at line~\ref{algo:compression:l4}.
Recall that, according to
Lemma~\ref{lemma:halt_compression}, the \texttt{while-loop}
at line~\ref{algo:compression:l2} iterates at most $d_{S,T}$ times.
We conclude that, in this case, the \texttt{compression\_phase()} halts within the following time bound:
\[t_\M\, d_{S,T} = O\left(\min\Big(\sqrt{|\F|\, d_{S,T}\, n}, |\F|\Big) |\F|^2\, d^2_{S,T}\, n \right).\]

    \item In the case where  \texttt{compute\_Lehman-Ron\_paths()} gets executed,
    which happens whenever $|\M|=|\F|+1$, we must now take its time complexity into account, which we analyze below.

% Now, let us analyze the time complexity of \texttt{compute\_Lehman-Ron\_paths()},
% which is invoked at line~\ref{algo:compression:l8} of the \texttt{compression\_phase()}
% whenever $|\M|=|\F|+1$.
First, consider the set $\M_\S$ computed at line~\ref{algo:compression:l6} of \texttt{compression\_phase()}.
The following bound holds on its size: \[|\M_\S| = |\M| = |\F| + 1.\]
The same bound holds for the set $\M_\T\subseteq \T$ which is computed at line~\ref{algo:compression:l7} --- namely:
$|\M_\T| = |\M| = |\F|+1$.
By \Cref{thm:algo_lehmanron}, provided that we consider the parameter $m=|\M|=O(|\F|)$,
invoking \texttt{compute\_Lehman-Ron\_paths()}
on input $\langle \M_\S, \M_\T, \M, n \rangle$ takes time at most $t_{\text{LR}}$, where:
$$t_{\text{LR}} = O\left( m^{5/2} n^{3/2} d_{S,T}  \right) = O\left( |\nonviableset|^{5/2} n^{3/2} d_{S,T} \right).$$

Recall that, by Lemma~\ref{lemma:halt_compression}, the \texttt{while-loop}
at line~\ref{algo:compression:l2} iterates at most $d_{S,T}$ times.
At each of such iterations, a brand new matching $\M$ gets computed at line~\ref{algo:compression:l4}.
Finally, at the very last of such iterations, provided that $|\M|>|\F|$, then the procedure
\texttt{compute\_Lehman-Ron\_paths()} is invoked at line~\ref{algo:compression:l8}.
Therefore, we conclude that whenever \texttt{compression\_phase()}
executes \texttt{compute\_Lehman-Ron\_paths()} at line~\ref{algo:compression:l8},
then it halts within the following time bound:
\begin{align*}
& t_\M  \, d_{S,T}  + t_{\text{LR}}  = \\
& O\left(\min\Big(\sqrt{|\F|\, d_{S,T}\, n}, |\F|\Big) |\F|^2\, d^2_{S,T}\, n + |\nonviableset|^{5/2} n^{3/2} d_{S,T} \right)
\end{align*}
\end{enumerate}
\end{proof}

\begin{proposition}\label{prop:complexity_solve_decision}
The \textsc{Decision-Task} of {\mainproblem} can be solved within the following time bound
on any input $\langle S, T, \F, n \rangle$:
\[
O\left(\min\Big(\sqrt{|\F|\, d_{S,T}\, n}, |\F|\Big) |\F|^2\, d^3_{S,T}\, n \right).
\]
\end{proposition}
\begin{proof}
Let us consider the procedure \texttt{solve\_\mainproblem()} of Algorithm~\ref{ALGO:solve}.
By \Cref{prop:solve_halts_finite}, the \texttt{while-loop} at line~\ref{algo:solve:l4} iterates at most $d_{S,T}$ times.
At each iteration, \texttt{double-bfs\_phase()} is invoked at line~\ref{algo:solve:l5},
and \texttt{compression\_phase()} is invoked soon after at line~\ref{algo:solve:l11}.
By \Cref{lemma:complexity_double-bfs_phase}, the most expensive one between the two procedures is clearly \texttt{compression\_phase()}.
Recall that, if we are content with solving the \textsc{Decision-Task} of {\mainproblem},
then the \texttt{compression\_phase()} can be implemented so that it always halts without
ever executing the  procedure \texttt{compute\_Lehman-Ron\_paths()} at line~\ref{algo:compression:l8}.
Therefore, by \Cref{lemma:complexity_compression_phase},
each invocation of \texttt{compression\_phase()} takes time at most
$$O\left(\min\Big(\sqrt{|\F|\, d_{S,T}\, n}, |\F|\Big) |\F|^2\, d^2_{S,T}\, n \right).$$
Since we have at most $d_{S,T}$ of such invocations, then the thesis follows.
\end{proof}

\begin{proposition}\label{prop:complexity_solve_search}
The \textsc{Search-Task} of {\mainproblem} can be solved within the following time bound
on any input $\langle S, T, \F, n \rangle$:
\[
O\left(\min\Big(\sqrt{|\F|\, d_{S,T}\, n}, |\F|\Big) |\F|^2\, d^3_{S,T}\, n +
 |\nonviableset|^{5/2} n^{3/2} d_{S,T} \right).
\]
\end{proposition}
\begin{proof}
Let us consider the procedure \texttt{solve\_\mainproblem()} of Algorithm~\ref{ALGO:solve}.
By \Cref{prop:solve_halts_finite}, the \texttt{while-loop} at line~\ref{algo:solve:l4} iterates at most $d_{S,T}$ times.
At each iteration, \texttt{double-bfs\_phase()} is invoked at line~\ref{algo:solve:l5},
and \texttt{compression\_phase()} is invoked shortly after at line~\ref{algo:solve:l11}.
By \Cref{lemma:complexity_double-bfs_phase}, the most expensive step between the two is clearly the \texttt{compression\_phase()}.
Recall that, if we aim to solve the \textsc{Search-Task} of \mainproblem, then the \texttt{compression\_phase()} possibly executes
the \texttt{compute\_Lehman-Ron\_paths()} procedure at line~\ref{algo:compression:l8}.
Nevertheless, whenever \texttt{compression\_phase()} executes \texttt{compute\_Lehman-Ron\_paths()} at line~\ref{algo:compression:l8},
then the procedure \texttt{solve\_\mainproblem()} halts shortly after at line~\ref{algo:solve:l12}.
This means that the only invocation of \texttt{compression\_phase()} that possibly
executes \texttt{compute\_Lehman-Ron\_paths()} is the very last invocation.
Then, each invocation of \texttt{compression\_phase()}, except the very last one,
halts within the following time bound by Lemma~\ref{lemma:complexity_compression_phase}:
$O\big(\min\big(\sqrt{|\F|\, d_{S,T}\, n}, |\F|\big) |\F|^2\, d^2_{S,T}\, n \big).$
Since the very last invocation of \texttt{compression\_phase()} possibly executes the procedure
\texttt{compute\_Lehman-Ron\_paths()} at line~\ref{algo:compression:l8}, the following time bound holds on the last invocation of
\texttt{compression\_phase()} by Lemma~\ref{lemma:complexity_compression_phase}:
\[O\left(\min\Big(\sqrt{|\F|\, d_{S,T}\, n}, |\F|\Big) |\F|^2\, d^2_{S,T}\, n + |\nonviableset|^{5/2} n^{3/2} d_{S,T} \right).\]
Since there are at most $d_{S,T}$ invocations of the \texttt{compression\_phase()}, the thesis follows.
\end{proof}

%% file: esa-guided.bbl
\begin{thebibliography}{19}
\providecommand{\natexlab}[1]{#1}
\providecommand{\url}[1]{\texttt{#1}}
\expandafter\ifx\csname urlstyle\endcsname\relax
  \providecommand{\doi}[1]{doi: #1}\else
  \providecommand{\doi}{doi: \begingroup \urlstyle{rm}\Url}\fi

\bibitem[B\'erard et~al.(2004)B\'erard, Bergeron, and
  Chauve]{Berard:Bergeron:Chauve:RECOMB:2004}
{\sc S.~B\'erard, A.~Bergeron, and C.~Chauve}, {\em Conservation of
  combinatorial structures in evolution scenarios}, in RECOMB'04, vol.~3388 of
  LNCS, Springer, 2004, pp.~16--19.

\bibitem[Bergeron et~al.(2004)Bergeron, Blanchette, Chateau, and
  Chauve]{Bergeron:Blanchette:Chateau:Chauve:WABI:2004}
{\sc A.~Bergeron, M.~Blanchette, A.~Chateau, and C.~Chauve}, {\em
  Reconstructing ancestral gene orders using conserved intervals}, in WABI'04,
  vol.~3240 of LNCS, Springer, 2004, pp.~14--25.

\bibitem[Comin et~al.(2016)Comin, Labarre, Rizzi, and Vialette]{Comin2016}
{\sc C.~Comin, A.~Labarre, R.~Rizzi, and S.~Vialette}, {\em Sorting with
  Forbidden Intermediates}, Algorithms for Computational Biology: Third
  International Conference, AlCoB 2016, Trujillo, Spain, June 21-22, 2016,
  Proceedings, Springer, 2016, pp.~133--144.

\bibitem[Cu\'enot(1905)]{Cuenot1905}
{\sc L.~Cu\'enot}, {\em Les races pures et leurs combinaisons chez les souris},
  Archives de Zoologie Experimentale, 3 (1905), pp.~cxxiii--cxxxii.

\bibitem[Deza and Huang(1998)]{Deza:1998}
{\sc M.~Deza and T.~Huang}, {\em {Metrics on permutations, a survey}}, Journal
  of Combinatorics, Information and System Sciences, 23 (1998), pp.~173--185.

\bibitem[Diestel(2005)]{Diestel2005}
{\sc R.~Diestel}, {\em Graph Theory (Graduate Texts in Mathematics)}, Springer,
  2005.

\bibitem[Fertin et~al.(2009)Fertin, Labarre, Rusu, Tannier, and
  Vialette]{Fertin2009}
{\sc G.~Fertin, A.~Labarre, I.~Rusu, E.~Tannier, and S.~Vialette}, {\em
  {Combinatorics of Genome Rearrangements}}, The MIT Press, 2009.

\bibitem[Figeac and Varr\'{e}(2004)]{Figeac:Varre:WABI:2004}
{\sc M.~Figeac and J.-S. Varr\'{e}}, {\em {Sorting by reversals with common
  intervals}}, in WABI'04, vol.~3240 of LNCS, Springer, 2004, pp.~26--37.

\bibitem[Gabow et~al.(1976)Gabow, Maheshwari, and Osterweil]{GMO76}
{\sc H.~Gabow, S.~Maheshwari, and L.~Osterweil}, {\em On two problems in the
  generation of program test paths}, IEEE Trans. Software Eng.,  (1976),
  pp.~227--231.

\bibitem[Gluecksohn-Waelsch(1963)]{Gluecksohn-Waelsch63}
{\sc S.~Gluecksohn-Waelsch}, {\em Lethal genes and analysis of
  differentiation}, Science, 142 (1963), pp.~1269--1276.

\bibitem[Goldreich et~al.(2000)Goldreich, Goldwasser, Lehman, Ron, and
  Samorodnitsky]{GoldsLehmanRon2001}
{\sc O.~Goldreich, S.~Goldwasser, E.~Lehman, D.~Ron, and A.~Samorodnitsky},
  {\em Testing monotonicity}, Combinatorica, 20 (2000), pp.~301--337.

\bibitem[Hopcroft and Karp(1973)]{HK73}
{\sc J.~Hopcroft and R.~Karp}, {\em An $n^{5/2}$ algorithm for maximum
  matchings in bipartite graphs}, SIAM Journal on Computing, 2 (1973),
  pp.~225--231.

\bibitem[Knuth(1973)]{DBLP:books/aw/Knuth73}
{\sc D.~E. Knuth}, {\em The Art of Computer Programming, Volume {III:} Sorting
  and Searching}, Addison-Wesley, 1973.

\bibitem[Krause et~al.(1973)Krause, Smith, and Goodwin]{KSG73}
{\sc K.~Krause, R.~Smith, and M.~Goodwin}, {\em Optimal software test planning
  through automated network analysis}, in Proceedings 1973 IEEE Symposium
  Computer Software Reliability, IEEE, 1973, pp.~18--22.

\bibitem[Labarre(2013)]{Labarre2013}
{\sc A.~Labarre}, {\em Lower bounding edit distances between permutations},
  SIAM Journal on Discrete Mathematics, 27 (2013), pp.~1410--1428.

\bibitem[Lakshmivarahan et~al.(1993)Lakshmivarahan, Jwo, and
  Dhall]{Lakshmivarahan1993}
{\sc S.~Lakshmivarahan, J.-S. Jwo, and S.~K. Dhall}, {\em {Symmetry in
  interconnection networks based on Cayley graphs of permutation groups: A
  survey}}, Parallel Computing, 19 (1993), pp.~361--407.

\bibitem[Lehman and Ron(2001)]{LehmanRon2001}
{\sc E.~Lehman and D.~Ron}, {\em On disjoint chains of subsets}, Journal of
  Combinatorial Theory, Series A, 94 (2001), pp.~399--404.

\bibitem[Logan and Shahriari(2004)]{LS04}
{\sc M.~Logan and S.~Shahriari}, {\em A new matching property for posets and
  existence of disjoint chains}, Journal of Combinatorial Theory, Series A, 108
  (2004).

\bibitem[Yinnone(1997)]{Yin97}
{\sc H.~Yinnone}, {\em On paths avoiding forbidden pairs of vertices in a
  graph}, Discrete Applied Mathematics, 74 (1997), pp.~85--92.

\end{thebibliography}
